\documentclass[pdflatex,sn-basic,Numbered,oneside]{sn-jnl}

\usepackage{amsmath,amssymb,mathtools}
\usepackage{booktabs}
\usepackage{array}
\usepackage{microtype}
\usepackage{tikz}
\usepackage{xcolor}
\usepackage[normalem]{ulem}
\usetikzlibrary{positioning,arrows.meta,calc}

\theoremstyle{thmstyleone}
\newtheorem{theorem}{Theorem}[section]
\newtheorem{lemma}[theorem]{Lemma}

\newtheorem{axiom}[theorem]{Axiom}
\newtheorem{grindrule}[theorem]{Rule}

\theoremstyle{thmstylethree}
\newtheorem{definition}[theorem]{Definition}

\newcommand{\snail}[1][1]{%
\begin{tikzpicture}[baseline=-0.55ex,scale=#1*0.105,line width=#1*0.5pt,
                    line cap=round,line join=round]
  \draw (-3.05,-1.15) .. controls (-3.95,-1.15) and (-4.15,-0.35) .. (-3.75,-0.05)
        .. controls (-3.45,0.18) and (-3.15,-0.15) .. (-2.95,-0.5)
        .. controls (-2.2,-1.35) and (2.0,-1.55) .. (2.9,-1.15) -- cycle;
  \draw (-3.72,-0.15) -- (-4.45,1.15);
  \draw (-3.15,-0.35) -- (-3.25,1.05);
  \fill (-4.45,1.15) circle (0.28);
  \fill (-3.25,1.05) circle (0.28);
  \draw plot[domain=0:900,samples=140,variable=\t,smooth]
    ({0.35+0.00245*\t*cos(\t)},{0.05+0.00245*\t*sin(\t)});
\end{tikzpicture}%
}

\newcolumntype{L}[1]{>{\raggedright\arraybackslash}p{#1}}
\newcommand{\finding}[1]{\emph{#1}\par\medskip}
\newcommand{\notclaimed}[1]{\par\smallskip\noindent\textit{Not claimed.} #1\par}

\ifdefined\showtrackedchanges
  \newcommand{\added}[1]{\textcolor{blue}{#1}}
  \newcommand{\deleted}[1]{\textcolor{red}{\sout{#1}}}
\else
  \newcommand{\added}[1]{#1}
  \newcommand{\deleted}[1]{}
\fi
\newcommand{\changegap}{\ifdefined\showtrackedchanges\space\fi}

\begin{document}

\title[AI Grinding for Fun and Cryptanalysis]{AI Grinding for Fun and
Cryptanalysis}

\author*[1,2]{\fnm{Lukasz} \sur{Olejnik}}\email{me@lukaszolejnik.com}
\author[3,4]{\fnm{Bartosz} \sur{Naskrecki}}\email{bartnas@amu.edu.pl}

\affil*[1]{\orgdiv{Department of War Studies},
\orgname{King's College London}}
\affil[2]{\orgname{Independent Researcher}}
\affil[3]{\orgdiv{Faculty of Mathematics and Computer Science},
\orgname{Adam Mickiewicz University}}
\affil[4]{\orgdiv{Centre for Credible AI},
\orgname{Warsaw University of Technology}}

\abstract{\deleted{AI may offer many ways of taking a construction apart, faster than any of them
can be checked. Exact arithmetic and a control with a predetermined outcome
decide which survive. Grinding of this kind is a young direction in
cryptanalysis, and what follows is its current yield.}

\deleted{AI systems can generate candidate attacks faster than they can be
assessed. We study an evidence-gated workflow that can operate autonomously
through hypothesis generation, exact testing, discriminating controls, and
recursive attacks on surviving limitations. Human review is delayed until a
candidate has produced a reproducible insight at the parameters under study.
The operator then verifies source fidelity, security consequence, scope,
provenance, and disclosure status.}

\deleted{The primary contribution is this separation between autonomous
discovery and human cryptanalytic certification. The attacks below serve as
case studies of the evidence gate in operation.}

\added{We present an autonomous cryptanalysis workflow in which agents
generate, test, and refine hypotheses before human review. The autonomous
stage returns reproducible candidate findings with exact witnesses, controls,
code, environments, and run records. A researcher then determines whether the
evidence establishes a break, defect, or coverage gap. The attacks below serve
as case studies of this method.}

Two failure modes recur, and naming them is more useful than listing the
targets, so they come first. In the first, a public algebraic map
\added{or an input representation} erases or exposes the relation a
construction must keep hidden. The maps are
unglamorous: multiplication by zero, the boundary coefficients of a polynomial
product, a quotient, a character, a Schur square, \added{or a variable-length
byte encoding without boundaries}. In the second, a
distribution is not the law it is claimed to be, whether a simulator, an error
law, or a parameter certification. Several targets fail on both axes at once,
which is what makes them decisive.

Every result carries an exact witness and at least one discriminating control,
and every stated boundary carries a proof. Three further targets yielded no
attack but a narrower guarantee than a generic reading suggests, and are
reported separately so that no coverage gap reads as a break.

Eight published constructions fail at parameters or claims their own authors
state. A Ring-LWR commitment opens to every message with probability one; two
middle-product encryption rows are decrypted from a single ciphertext; a
lattice e-voting protocol loses receipt-freeness, so a coerced voter can no
longer hide how they voted; a published
permutation recovery against an updatable encryption scheme lifts by linear
algebra to the old decryption key; an explicit normal basis collapses a
degree-63 instance into seven degree-nine ones; a signature hash outside the
lattice setting maps two printable messages of equal length to one digest; and
a rerandomisable scheme's accept bit is a threshold oracle on its own
decryption noise. Kept separate on purpose, a group-ring decision claim and a
multivariate MinRank hardening fail as assumption and accounting defects rather
than as broken constructions. Every failure sits one level above the assumption
it rests on. Who knows what else awaits out there?}


\maketitle

\raggedbottom

\section{Grinding: a direction for
machine-assisted cryptanalysis}
\label{sec:grind}

AI may offer many ways of taking a construction apart. A model reads a paper
and returns candidate lines of attack in quantity, and almost none of them
survive contact with an exact test. Deciding which candidate is real is the
expensive part, and that is the subject here.

Grinding is the name for the arrangement that answers it. Candidate generation
is treated as deliberately high in volume and low in precision, which is
affordable only because the workflow never treats generator output as evidence.
Exact arithmetic, a discriminating control, and the source document decide what
is admitted. Whatever survives is pushed until it either breaks something the
paper promises or runs into a wall that can be proved to be a wall.
Cryptanalysis has always rewarded depth, one person carrying one idea a long
way, and nothing here replaces that.

The snail is an emblem, not a claim about pace. What it stands for is the
trail: ground covered one claim at a time, in the open, where anyone can
retrace it. Whether this amounts to a new direction in cryptanalysis and in
mathematics more broadly is an open question. What follows is the yield and the
method that produced it.

The blunt version: pick a paper, read what it actually claims, attack every
step where the argument moves from one setting to another, turn each suspicion
into a precise statement, test that statement on exact small instances,
discard bad ideas quickly, and push whatever survives until it either stops for
a reason that can be written down or breaks the advertised security property.
Then check everything against the source and record what broke and what did
not.

\begin{figure}[!ht]
\centering
\resizebox{\linewidth}{!}{%
\begin{tikzpicture}[
  font=\small,
  box/.style={draw, rounded corners=2pt, align=center, inner sep=5pt},
  plain/.style={align=center, inner sep=3pt},
  >={Stealth[round,length=5pt]},
  every path/.style={thick}
]
\node[plain] (paper) {\textbf{PAPER}};
\node[box, below=5mm of paper, text width=5.0cm] (gen)
  {AI dreams up a thousand\\ways it could break};
\node[box, below=5mm of gen] (test) {\textbf{EXACT TEST}};
\node[plain, below left=9mm and 12mm of test] (fails) {FAILS};
\node[plain, below right=9mm and 12mm of test] (works) {WORKS};
\node[box, below=5mm of works] (push) {\textbf{PUSH MORE}};
\node[box, below=5mm of push, text width=5.0cm] (more)
  {bigger consequence?\\proof $\rightarrow$ attack $\rightarrow$ key};
\node[box, below=5mm of more] (check)
  {\deleted{\textbf{CHECK ANALYSIS AND STUFF}}
   \ifdefined\showtrackedchanges\\[-1pt]\fi
   \added{\textbf{HUMAN AUDIT}}};

\draw[->] (paper) -- (gen);
\draw[->] (gen) -- (test);
\draw[->] (test) -- (fails);
\draw[->] (test) -- (works);
\draw[->] (works) -- (push);
\draw[->] (push) -- (more);
\draw[->] (more) -- (check);

\draw[->] (fails.west) -- ++(-9mm,0) |- (gen.west);
\draw[->] (check.east) -- ++(24mm,0) |- (gen.east);
\end{tikzpicture}}
\caption{Most ideas fail. Survivors are pushed until they either hit a proved
boundary or break the advertised security property. Everything that survives
is rechecked against the exact paper version and an exact calculation or
executable witness. \deleted{The autonomous path reaches a candidate insight;
the final box is the human promotion boundary.}\changegap
\added{The autonomous path produces a reproducible candidate result; human
review determines its cryptanalytic status.}}
\label{fig:loop}
\end{figure}
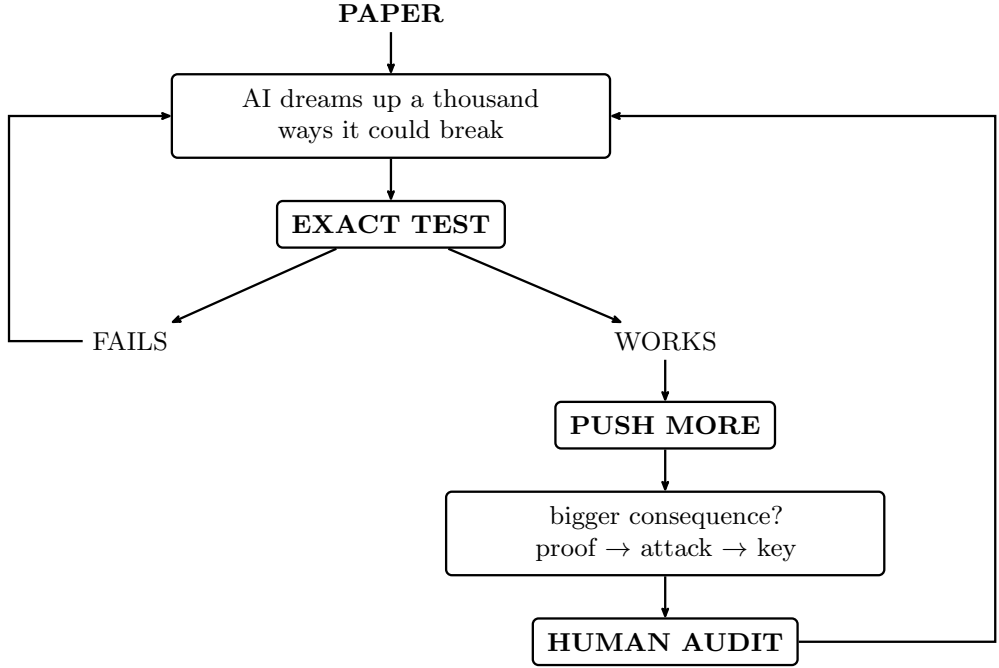

The push is where the results come from. Auditing a proof once is ordinary.
Every surviving result carries a stated limitation, and that limitation is the
next thing to attack:

\[
\begin{gathered}
  \text{idea}\;\rightarrow\;\text{exact witness}\;\rightarrow\;
  \text{``only works when \dots''}\\
  \rightarrow\;\text{attack that limitation}\;\rightarrow\;
  \text{stronger result}.
\end{gathered}
\]

That move has two productive exits, so it is never wasted. A boundary that
moves gives a stronger attack. A boundary that refuses to move becomes a
maximality statement, which is also a result. Six of the findings below reach
the top level and break an advertised security property; the rest stop lower,
and Table~\ref{tab:trace} records where each one stopped and why. Nobody camped
on level one for lack of trying. This is the operation that separates the
method from ordinary machine-assisted auditing. Auditing asks
whether a candidate flaw is real and stops when it has an answer. Grinding
treats every verified limitation of a real flaw as the next candidate, so a
confirmed result is an input rather than an endpoint. The arrangement works as a game
with levels, which is what Figure~\ref{fig:ladder} shows: a finding enters at
the left and climbs for as long as the evidence will carry it.

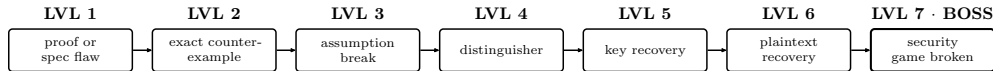
\begin{figure}[!ht]
\centering
\resizebox{\linewidth}{!}{%
\begin{tikzpicture}[
  font=\footnotesize,
  rung/.style={draw, rounded corners=2pt, align=center, inner sep=4pt,
               minimum height=8mm, text width=2.05cm},
  boss/.style={rung, line width=1.1pt},
  lvl/.style={font=\scriptsize\bfseries, inner sep=1.5pt},
  who/.style={align=center, font=\scriptsize},
  >={Stealth[round,length=4pt]},
  every path/.style={thick}
]
\node[rung] (a) {proof or spec flaw};
\node[rung, right=3.5mm of a] (b) {exact counter-example};
\node[rung, right=3.5mm of b] (c) {assumption break};
\node[rung, right=3.5mm of c] (d) {distinguisher};
\node[rung, right=3.5mm of d] (e) {key recovery};
\node[rung, right=3.5mm of e] (f) {plaintext recovery};
\node[boss, right=3.5mm of f] (g) {security game broken};
\foreach \x/\y in {a/b,b/c,c/d,d/e,e/f,f/g} {\draw[->] (\x) -- (\y);}

\node[lvl, above=1.1mm of a] {LVL 1};
\node[lvl, above=1.1mm of b] {LVL 2};
\node[lvl, above=1.1mm of c] {LVL 3};
\node[lvl, above=1.1mm of d] {LVL 4};
\node[lvl, above=1.1mm of e] {LVL 5};
\node[lvl, above=1.1mm of f] {LVL 6};
\node[lvl, above=1.1mm of g] {LVL 7 $\cdot$ BOSS};

\end{tikzpicture}}
\caption{The game. Every finding enters at level one and climbs for as long as
the evidence carries it, and it stops only when the next level cannot be
reached: either the bridge to a real consequence fails, or the wall ahead is
proved to be a wall. Table~\ref{tab:trace} records the level each finding here
reached.}
\label{fig:ladder}
\end{figure}

Generation needs direction, and the direction has two parts. The first is
where to look. A paper is attacked at its joins: a change of algebra,
\added{a representation or encoding change,} a
distribution carried through a map, a rounding or CRT conversion, a quotient,
an exceptional value such as zero or a nonunit, a quantifier that widens from
one object to many, and a parameter substituted into a script instead of into
the theorem it came from. Every finding below sits on one of those, and
assembling that list is the cheap part of the work.

The second part is what to ask on arrival, and two questions carry almost all
of it. Which efficiently computable map survives this join with useful signal
still in it, and which distribution actually results from it. Those two
questions are the axes of Section~\ref{sec:axes}, which is why that section
appears before the findings and not as a summary of them. Both are cheap to
ask, neither requires understanding a construction in full, and both can be
put to a generator that has read one paper and nothing else.

\begin{table}[!t]
\centering
\footnotesize
\begin{tabular}{@{}L{2.15cm}L{2.3cm}c L{2.5cm}L{2.6cm}c@{}}
\toprule
Target & Join attacked & Axis & First survivor & Limitation pushed & Level \\
\midrule
Ring-LWR &
Exceptional value inside an accepted set &
A &
$d=0$ opens the zero commitment to every message &
Binding only, so the three protocols were attacked next &
7 \\
\addlinespace
Spinel &
Representation change from bytes to walk symbols &
A &
Variable-length codewords concatenated without boundaries &
One pair only, so count the bytes and lengths it reaches &
7 \\
\addlinespace
Middle-product &
Ordinary product used where a wrapping one is assumed &
A, B &
Boundary coefficients peel the masks at $t=9$ &
Published bound vacuous from $t=16$, so branch instead of abort &
7 \\
\addlinespace
CRT-RLWE reaction &
Exact noise extraction meeting a validity signal &
A &
Coherent shifts cancel until the centre lift wraps &
One coefficient, so binary-search the whole polynomial &
7 \\
\addlinespace
MinRank hardening &
Rank raised by perturbation, then averaged &
A &
The $aa$ block is publicly computable &
Average is not the attacker's direction, so take the tail &
1 \\
\addlinespace
E-voting &
Distribution carried through a simulator &
B &
Simulator marginal differs from the genuine law &
Distance looked parameter-dependent, so bound it over all $p$ &
7 \\
\addlinespace
Hollow-LWE &
Quotient, on top of a published recovery &
A, B &
$W_{ij}V_{ij}^{-1}=d_id_{I_j}$ &
Signs fixed only per component, so use the public key relation &
7 \\
\addlinespace
Function-field &
CRT decomposition of an explicit basis &
A, B &
Galois orbit supported factorwise &
Projected rate might rise, so compute it &
5 \\
\addlinespace
Group ring, full &
Quantifier widened from one ring to a family &
A &
Augmentation is a public map onto $\mathbb Z$ &
One functional only, so count every sign character &
3 \\
\addlinespace
Group ring, quotients &
Quotient of the selected families &
A, B &
Rank-four maps onto $M_2(\mathbb F_q)$ &
Only the quotient image, so try to widen it &
3 \\
\addlinespace
Hollow-LWE parameters &
Parameter put into a script, not into its theorem &
B &
Estimator invoked at $k$, not $k-h$ &
Says nothing about the true level, so cost the repair &
1 \\
\bottomrule
\end{tabular}
\caption{The trace each result actually followed, from the join that was
attacked to the level of Figure~\ref{fig:ladder} at which it stopped. The
limitation column is the push: in every row it is a statement the first
survivor could not support, taken as the next hypothesis. Two of these pushes
ended in a proved boundary instead of a stronger attack, which
Lemma~\ref{l:twoexits} counts as the second productive exit.}
\label{tab:trace}
\end{table}

Table~\ref{tab:trace} records the trace each result followed, so that the
findings can be read as instances of the method rather than as work labelled
with it afterwards. Section~\ref{sec:mp} is the clearest case of pushing. Stopping at ``the attack
works at $t=9$ and the theorem starts at $t=98$'' would have been correct and
uninformative. Attacking that boundary showed first that the published failure
bound measures the wrong quantity and the attack in fact runs to $t=24$, and
then that $98$ is essentially the whole-ciphertext entropy threshold $96.32$
plus the security slack $2\lambda/(k+1)$. The requirement is not an artefact
of one reduction; it is the point at which the ciphertext stops determining
the masks.

\paragraph{\texorpdfstring{\deleted{What the machine does, and what it does
not.}\changegap\added{Autonomous search and human review.}}{Autonomous search
and human review.}}
A large language model reads the sources, enumerates candidate paths from a
claim to a consequence, proposes witnesses, and writes the checking code. Such
proposals are unreliable in a specific and predictable way: they are fluent,
frequently wrong in the direction of overclaiming, and they do not separate a
proof gap from a break unless forced to. The working rule is that no proposal
counts as evidence until it is reduced to an exact statement a short
deterministic program can check, alongside a control whose outcome was fixed
before it ran. Several
plausible candidates fail there, and one appears in
Section~\ref{sec:group} as a rejected strengthening because its failure is
informative. The division holds throughout: the machine supplies breadth of
conjecture, the exact machinery decides admissibility, and consequence, scope,
and provenance stay with the researcher.

\begin{center}
\deleted{\emph{The generator proposes. The admission boundary certifies.}}

\deleted{\emph{The generator proposes. The evidence gate certifies the
computation.}}
\ifdefined\showtrackedchanges\\[-1pt]\fi
\deleted{\emph{The human confirms the cryptanalytic claim.}}
\ifdefined\showtrackedchanges\\[-1pt]\fi
\added{\emph{Autonomous search proposes and tests candidate findings; human
review determines their cryptanalytic status.}}
\end{center}

\noindent
\deleted{The validity of the second stage is deliberately independent of the reliability
of the first, and this is what keeps the method from ageing with a model
generation. Nothing in Appendix~\ref{app:calculus} requires the generator to be
calibrated, trustworthy, or even describable as a probability distribution. It
is required to produce candidates in quantity and nothing else, so a different
generator, or none at all, changes the yield without touching the argument that
licenses a result.}\changegap
\added{Admission depends on exact checks and predetermined controls rather
than on generator confidence. Changing the generator therefore affects search
yield, but not the admission criterion.} \deleted{Machine assistance makes it cheap to ask the two
questions across many papers at once and contributes nothing to answering
them.} \deleted{Machine assistance makes it cheap to pose the two questions
across many papers and can autonomously propose and execute candidate answers.
It does not by itself establish source fidelity, cryptanalytic consequence,
scope, or provenance; those obligations remain for the subsequent human
audit.}\changegap
\added{Machine assistance makes it possible to pose the two questions across
many papers and test candidate answers autonomously.}
Appendix~\ref{app:calculus} states the same arrangement as a calculus.

\deleted{The autonomous stage ends at a handoff boundary, not at a
cryptanalytic conclusion. It produces a candidate bundle containing the exact
claim, source version, witness, precommitted control, code, execution
environment, run record, limitations, and preliminary result type. The system
may label this a reproducible candidate insight at the stated parameters; only
after the operator audits the bundle may it be promoted to a cryptanalytic
break, defect, or coverage gap.}\changegap
\added{The autonomous stage returns the claim under test, source version,
witness, control, code, execution environment, run record, known limitations,
and proposed classification. That handoff is a reproducible candidate result,
not a classified one.}

Three disciplines do the filtering. Every equality, modular, integer,
encoding, and bound predicate is evaluated in exact rational or integer
arithmetic; the binomial tails in Section~\ref{sec:voting} are exact rationals
over $27^{4096}$, not normal approximations. Every executable check carries a
control whose expected outcome is fixed in advance, listed in
Table~\ref{tab:repro}. Breaks, assumption
defects, and coverage gaps are recorded as different things and never merged.

One failure shows what the discipline exists to catch. The first end-to-end
run of the group-ring distinguisher did not separate genuine from uniform
samples: at a convenient small modulus the union
bound evaluates to $1.78$, so no separation is available there. The analysis
was right and the run did not support it. Rerunning at the parameter sequence
the analysis actually specifies gives complete separation. Any reported
reproduction not run at the stated parameters should be treated as unverified.

Costs quoted below are arithmetic operation counts, not wall-clock benchmarks,
unless stated otherwise.

A procedure of this kind invites one obvious objection, and it is the reason
Section~\ref{sec:scope} exists. Output consisting entirely of breaks would be
indistinguishable from a selection procedure with an unknown filter on it,
particularly when the generator is fluent and inclined to overclaim. The answer
is to record the whole distribution of outcomes and keep the classes apart:
attacks, assumption defects, coverage gaps, results belonging to other authors,
paths left unresolved, and strengthenings that were tried and failed. Three
targets produced no attack and are reported as such. One strengthening was
refuted and is kept for the boundary its refutation proved.

Which leaves the claim the arrangement rests on. Hypothesis generation and
cryptanalytic evidence do not need the same reliability. Generation can be made
deliberately high in volume and low in precision, provided an exact and
adversarially controlled admission boundary sits between a proposal and
evidence, and provided the limitations of every survivor are attacked in turn.
\added{Generator quality determines search yield; exact checks and controls
determine what survives.} The machine does not have to be reliable for the
research to be.

\section{Results at a glance}

Table~\ref{tab:abbrev} expands the abbreviations used throughout, and
Table~\ref{tab:glance} lists every result with its status and the number that
carries it. The derivations follow in the same order.

\begin{table}[!t]
\centering
\footnotesize
\begin{tabular}{@{}lL{4.0cm}lL{4.0cm}@{}}
\toprule
\multicolumn{4}{@{}l}{\emph{Assumptions and problems}}\\
LWE & Learning With Errors &
RLWE & Ring Learning With Errors \\
MLWE & Module Learning With Errors &
LWR & Learning With Rounding \\
MP-LWE & Middle-Product LWE &
MP-CLWR & Middle-Product Computational LWR \\
LPN & Learning Parity with Noise &
SIS & Short Integer Solution \\
HFE & Hidden Field Equations &
MinRank & the problem of finding a low-rank combination of public matrices \\
\addlinespace
\multicolumn{4}{@{}l}{\emph{Constructions and notions}}\\
PKE & Public-Key Encryption &
UPKE & Updatable Public-Key Encryption \\
ABE & Attribute-Based Encryption &
AEAD & Authenticated Encryption with Associated Data \\
IND-CPA & Indistinguishability under Chosen-Plaintext Attack &
IND-CR-CPA & the corresponding notion for ciphertexts predating a key update \\
\addlinespace
\multicolumn{4}{@{}l}{\emph{Other}}\\
CRT & Chinese Remainder Theorem &
SLH-DSA & Stateless Hash-Based Digital Signature Algorithm \\
FORS & Forest Of Random Subsets, a component of SLH-DSA &
WOTS+ & Winternitz One-Time Signature Plus \\
\bottomrule
\end{tabular}
\caption{Abbreviations, expanded here and at first use in the text.}
\label{tab:abbrev}
\end{table}

\begin{table}[!t]
\centering
\footnotesize
\begin{tabular}{@{}L{2.7cm}L{2.7cm}L{2.0cm}L{4.35cm}@{}}
\toprule
Target & Claim defeated & Status & Evidence \\
\midrule
Ring-LWR commitment and proofs~\cite{tmm2025} &
Binding, and soundness of all three protocols &
Break, probability one &
Zero multiplier opens one commitment to every message; extraction is
undefined for half the challenge pairs \\
\addlinespace
Spinel signature hash~\cite{spinel2026} &
Collision and second-preimage resistance of the hash &
Break, deterministic, no search &
Distinct messages share one walk string; $244$ of $255$ nonzero bytes have an
immediate second preimage \\
\addlinespace
Two-limb CRT-RLWE rerandomisation~\cite{eprint20261618} &
Composed malicious-rerandomiser claim &
Break, conditional on an observable reaction &
Coherent CRT shifts make the accept bit a noise threshold; $131072$ reactions
give the key \\
\addlinespace
HFE with $LL'$ perturbations~\cite{eprint2026404} &
The MinRank hardening behind the parameter tables &
Accounting and transcript defects &
Ambient dimension publicly reducible, averaged rank attacker-chosen, and a
degree-2 transcript invariant \\
\addlinespace
Middle-product PKE, $t=9$ rows~\cite{rsss2017,bai2019} &
IND-CPA confidentiality &
Break, passive, one ciphertext &
All $2313$ masks and the plaintext recovered, success
$\ge0.999999730074$; range extends to $t=24$ \\
\addlinespace
Lattice e-voting receipts~\cite{fpsw2025} &
Receipt-freeness and vote privacy against a transcript coercer &
Break, unconditional &
Simulator distance $\ge1/15$ for every permitted parameter; test advantage
$0.999999804160$, and provably optimal \\
\addlinespace
Hollow-LWE UPKE~\cite{abl2025} &
IND-CR-CPA confidentiality of pre-update ciphertexts &
Break, given the published permutation recovery &
A valid old key recovered and the challenge decrypted $12/12$ at three
dimensions, in $O(k^3+nk^2)$ operations \\
\addlinespace
Function-field Normal Ring-LPN~\cite{bcd2022} &
Search hardness of the explicit constructive basis &
Break of the printed instance &
$X^{63}+X^7+1$ splits into seven degree-nine factors; $3584$ candidate tests
replace $2^{63}$ \\
\addlinespace
Semidirect group-ring LWE~\cite{lf2026} &
Decision hardness, full rings and selected quotients &
Assumption defect &
Two samples distinguish with false acceptance $64N^{-7}\log^2N$; quotient
image shown maximal \\
\addlinespace
Hollow-LWE parameters~\cite{abl2025} &
The printed 128, 192 and 256-bit labels &
Certification defect &
The estimator is called at dimension $k$, not $k-h$; repairing both defects
costs $1.24$ to $1.54\times$ \\
\addlinespace
Hint-MLWE~\cite{klss2023} &
The unrestricted parameterised definition &
Assumption defect &
A permitted degenerate choice gives advantage $1-q^{-nm}$ \\
\bottomrule
\end{tabular}
\caption{Every result, its status, and the evidence carrying it.
Section~\ref{sec:scope} adds three targets that produced no attack.}
\label{tab:glance}
\end{table}

\paragraph{How results are classified.}
Three categories are kept apart throughout, because merging them would
misstate the evidence. A \emph{break} reaches a security property the source
claims: key or plaintext recovery, binding, soundness, or a named privacy
notion. An \emph{assumption or certification defect} shows that a stated
claim is false or unsupported without exhibiting an attack on a construction.
A \emph{coverage gap}, collected in Section~\ref{sec:scope}, is a reading of
what a source actually proves, not a defect in it. Results attributable to
other authors are marked at the point of use and are never counted here.

\section{Two axes: what to ask at each join}
\label{sec:axes}

The two axes below are search directions before they are a classification.
They are the questions put at every join, and they are what a generator is
pointed at when it reads a construction for the first time. The constructions
analysed here come from different lines of work and rest on different
assumptions, and the same two questions apply to all of them without
modification, which is the property that makes the search worth running at
volume.

The axes also classify every defect found, and that is worth stating carefully.
A search that only asks two questions will only return answers to those two
questions, so the classification is a property of the search and not a
discovery about lattice cryptography. What the classification does establish is
narrower and still useful: neither question was exhausted by a single target,
and both kept producing across unrelated constructions.

\paragraph{Axis A: a public map that should not have survived.}
A construction hides a relation and publishes objects derived from it. If some
efficiently computable map sends the published objects somewhere the hidden
relation is either trivial or enumerable, the assumption is not doing the work
it appears to do. The maps here are unglamorous: multiplication by
zero, reading the first and last coefficients of a polynomial product,
reduction modulo a group-theoretic ideal, summing group-ring coefficients, and
squaring a code coordinatewise. \added{The same axis includes representation
changes: a non-uniquely-decodable input encoding identifies distinct messages
before the advertised algebraic problem is reached.} None requires lattice
reduction.

\paragraph{Axis B: a distribution that is not the claimed law.}
A construction, a simulator, or a parameter script asserts that some object
follows a distribution. If the asserted law and the actual law differ, the
consequences range from a distinguisher to an invalid security label. The
mismatches here are a simulator whose marginal is wrong at every permitted
parameter, an error law that a decision definition requires but never supplies,
a projected noise rate that fails to rise when the dimension drops, and a
lattice estimator invoked at the wrong source dimension.

Table~\ref{tab:axes} places each target on both axes. Several sit on both,
which is why they are decisive: a public map exposes an object, and the
distribution of that object is not what the proof needs.

\begin{table}[t]
\centering
\small
\begin{tabular}{@{}L{3.6cm}L{4.3cm}L{4.3cm}@{}}
\toprule
Target & Axis A: the public map & Axis B: the distribution \\
\midrule
Ring-LWR commitment &
Multiplication by the permitted $d=0$ &
Not load-bearing here \\
Spinel signature hash &
Byte-to-ternary encoding, non-injective under concatenation &
Not load-bearing here \\
Two-limb CRT-RLWE &
Coherent shift surviving the CRT recombination &
Centre-lift wrap moves the plaintext by exactly one \\
HFE with $LL'$ perturbations &
Principal $aa$ block, and the kernel of a transcript relation &
Differential rank taken at its average, not its tail \\
Middle-product PKE, $t=9$ &
Boundary coefficients of an ordinary polynomial product &
Binary masks make each layer a $2^t$ subset sum \\
Lattice e-voting receipts &
Public special decryption of the claimed vote &
Simulator marginal differs from the genuine convolution law \\
Hollow-LWE UPKE &
Schur square, then a signed-graph traversal &
Estimator invoked at dimension $k$ instead of $k-h$ \\
Function-field Normal Ring-LPN &
Projection onto one CRT component &
Projected Bernoulli rate unchanged, not amplified \\
Semidirect group-ring LWE &
Augmentation, and rank-four quotients &
Decision definition supplies a set, not a distribution \\
\bottomrule
\end{tabular}
\caption{Each target placed on both axes.}
\label{tab:axes}
\end{table}

\section{Findings}

Findings are ordered by demonstrated impact: affected scope, consequence,
centrality to the assurance claim, and evidence strength.

\subsection{A Ring-LWR commitment opens to every message}

\finding{The commitment verifier accepts an opener-chosen multiplier whose
permitted values include zero, so one commitment opens to two distinct
messages for every public key.}

The construction of Talapatra, Mishra, and Mukhopadhyay~\cite{tmm2025} commits
to $m$ by hashing and rounding,
\[
  \mathbf c=\operatorname{Round}_{q\rightarrow p_1}
       \bigl(\mathbf aH(\operatorname{enc}(m)\mathbin\|\operatorname{enc}(r))\bigr).
\]
The honest opening uses multiplier $1$, but verification accepts any $(m,r,d)$
with
\[
  d\,\operatorname{Round}_{q\rightarrow p_1}
       \bigl(\mathbf aH(\operatorname{enc}(m)\mathbin\|\operatorname{enc}(r))\bigr)=\mathbf c,
  \qquad \lVert d\rVert_\infty\leq1,
  \qquad \deg d<n/2 .
\]
The accepted set contains $d=0$. Taking $\mathbf c=0$ and $d_0=d_1=0$ makes
both verification equations read $0=0$ for any two distinct messages,
independently of the hash outputs and the public key. Binding fails with
probability one at every valid parameter set.

The same cancellation defeats the proof system. In the preimage protocol, a
prover sends $\mathbf c_{\mathrm{aux}}=0$, receives the challenge $\delta$, and
answers with $\mathbf t=-\delta\mathbf c$ and $s_m=0$, opening the auxiliary
commitment with multiplier zero. Both verifier equations become $0=0$ for every
challenge, so the failure is in soundness, not in extraction
alone. The two derived protocols are undefined as printed, since they add
elements of $R_p^k$ to elements of $R_{p_1}^k$; after the natural type
correction they accept explicit false statements under the same strategy.
Independently, the special-soundness extractor divides by challenge
differences such as $x-1$, which is a zero divisor in
$\mathbb F_2[x]/(x^{512}+1)$ because $x^{512}+1=(x+1)^{512}$ there. That
example is not isolated. Reduction modulo two lands in the local ring
$\mathbb F_2[x]/((x+1)^{512})$, so an element is a unit exactly when its
modulo-two weight is odd; challenges are binary, so the difference is a unit
exactly when the Hamming distance is odd. The fraction of distinct challenge
pairs at even distance is $(2^{255}-1)/(2^{256}-1)=\tfrac12-1/(2^{257}-2)$, so
an extractor drawing two independent uniform challenges divides by a nonunit
with probability one half. That is a constant, not a negligible quantity, and
repetition does not repair it.

The predecessor Ring-LWE commitment~\cite{bklp2015} inverts its opening
factor and therefore excludes zero. Ring-LWR hiding is untouched by this
attack; what fails is binding, and the soundness of everything built on it.

\subsection{A signature hash collides before its hard problem is reached}
\label{sec:spinel}

\finding{An encoding that converts bytes to walk symbols by suppressing leading
zeros is not injective, so two distinct messages drive the same algebraic walk
and produce the same digest with no search.}

Spinel~\cite{spinel2026} keeps the SLH-DSA architecture and replaces its hash
with an algebraic one, whose security rests on navigating a Cayley graph of
$\mathrm{SL}_4(\mathbb F_p)$. That problem is never reached. Each byte
$b\in\{0,\ldots,255\}$ is written in base $3$, which needs six trits since
$3^6=729>256$; the specification then suppresses the leading zeros, maps each
remaining trit $t$ to $t+1$, and concatenates the resulting codewords $C(b)$
into the string $E(M)=C(b_1)\Vert\cdots\Vert C(b_m)$ that drives the walk.

Every byte does have a unique minimal representation, and that is the step the
argument turns on: the source transfers the underlying hash's security on the
ground that it applies only one-to-one encodings to the input and output. Per
byte the code is indeed injective. A concatenation of variable-length
codewords without boundaries is not. Since $\mathtt{0x01}=(1)_3$ and
$\mathtt{0x04}=(11)_3$,
\[
  E(\mathtt{01}\,\mathtt{04})=2\Vert22=222=22\Vert2=E(\mathtt{04}\,\mathtt{01}),
\]
so the walk receives one string for two messages and every later matrix
operation agrees. The witness is printable and length-preserving as well:
$\mathtt{0x28}=(1111)_3$ and $\mathtt{0x79}=(11111)_3$ give
$E(\texttt{"(y"})=E(\texttt{"y("})=2^9$. No group element is computed to find
either pair, and neither depends on $p$, on the generators, or on the output
serialisation.

Second preimages follow with the target fixed first. The codeword for
$\mathtt{0x04}$ is $22$, and $\mathtt{01}\,\mathtt{01}$ encodes to $2\Vert2$,
so $\mathtt{01}\,\mathtt{01}$ is a second preimage of the one-byte message
$\mathtt{04}$. A nonzero byte admits this splitting unless its minimal
representation has a single nonzero trit, which happens only for the eleven
values $d\cdot3^j$ with $d\in\{1,2\}$, so $244$ of the $255$ nonzero values
have an immediate second preimage. Length is not a defence either. Exhaustive
enumeration of the $255^2=65025$ two-byte strings over nonzero bytes gives
$49963$ distinct encodings, with $26730$ inputs in nontrivial collision
classes; reading a uniform $64$-byte value as $32$ two-byte blocks, at least
one block admits a same-length replacement \deleted{with probability}\changegap
\added{with probability at least}
$1-(1-26730/65536)^{32}=0.999999948$. Multicollisions are immediate, since
$\texttt{"(y"}$ and $\texttt{"y("}$ may be substituted independently at $r$
positions to give $2^r$ equal-length messages with one digest.

\added{The source instantiates every call to the tweakable hash
$T_h(P,T,M)=H(P\Vert T\Vert M)$ by passing the byte string $B=P\Vert T\Vert M$
to its hash, so $E(M_0)=E(M_1)$ gives
\[
  H(P\Vert T\Vert M_0)=H(P\Vert T\Vert M_1)
\]
for every fixed prefix and tweak.

Whether that reaches a signature forgery turns on one function the source does
not pin down. In SLH-DSA the message digest $H_{\mathrm{msg}}$ is not a
tweakable-hash call, and the specification here fixes only the tweakable hash.
If $H_{\mathrm{msg}}$ is instantiated with the same primitive, a chosen-message
adversary requests a signature on $M_0$ and returns it unchanged with $M_1$:
verification derives the same digest and walks the same FORS, WOTS+, and
hypertree path, so the substitution succeeds with probability one. If
$H_{\mathrm{msg}}$ retains a conventional hash, the two messages give different
digests and different indices, and this substitution fails. The collision and
second-preimage results above hold either way.}

The repair is the construction the paper states before it optimises. Retaining
all six trits fixes the boundaries and restores injectivity, and any uniquely
decodable byte code does the same. Prefixing the total message length does
not, because the witnesses above already have equal lengths.

A second and unrelated defect sits in the parameter argument. The digest is
sixteen $32$-bit words holding the entries of a matrix of
$\mathrm{SL}_4(\mathbb F_p)$ with $p=2^{31}-1$, described as $2^{512}$ possible
values. The top bit of each word is always zero and the determinant is
constrained, so with
$|\mathrm{SL}_4(\mathbb F_p)|=p^{6}(p^{2}-1)(p^{3}-1)(p^{4}-1)$ the digest
space is $\log_2|\mathrm{SL}_4(\mathbb F_p)|=465.000$ bits, and the idealised
classical birthday scale is $2^{232.5}$. Repairing the encoding does not change
this.

\par\smallskip\noindent
\deleted{\textit{Not claimed.} First-preimage resistance. Writing
$H=\Phi\circ E$, what fails is the injectivity of $E$; given an arbitrary
target matrix an attacker must still find a walk reaching it, and the ambiguity
only supplies alternative byte representations of a preimage already found.
Nor is the $\mathrm{SL}_4(\mathbb F_p)$ navigation problem attacked,
classically or quantumly: the collision is available before the walk begins,
so a quantum adversary gains nothing here. A signature forgery by message
substitution follows from the written design if the Spinel primitive
instantiates the message digest function, but no implementation was available
to test $\mathrm{Verify}(pk,M_1,\sigma)$ for a signature on a colliding $M_0$,
so no forgery is claimed.}\changegap
\added{\textit{Not claimed.} First-preimage resistance: what fails is the
injectivity of $E$, and given an arbitrary target matrix an attacker must still
find a walk reaching it. Nor is the $\mathrm{SL}_4(\mathbb F_p)$ navigation
problem attacked, classically or quantumly. The message substitution above is
claimed only under the stated condition on $H_{\mathrm{msg}}$, and no
implementation was available to test it.}
\par

\subsection{Middle-product encryption is decrypted from one ciphertext}
\label{sec:mp}

\finding{At both published $t=9$ parameter sets, the boundary coefficients of
the first ciphertext component recover all $2313$ binary encryption masks, and
the plaintext follows by public computation.}

The MP-LWE scheme of Ro\c{s}ca, Sakzad, Stehl\'e, and
Steinfeld~\cite{rsss2017} and the MP-CLWR scheme of Bai et al.~\cite{bai2019}
both publish
\[
  c_1=\sum_{i=1}^{t}r_i a_i\pmod q,
\]
where each $r_i$ has binary coefficients. This is ordinary, non-cyclic
polynomial multiplication, which is triangular at both ends. The constant
coefficient of $c_1$ depends only on the constant coefficients of the $r_i$;
once those are removed, the next coefficient exposes the next layer. After
subtracting the earlier layers, layer $\ell$ satisfies
\[
  T_\ell=\sum_{i=1}^{t}a_{i,0}r_{i,\ell}\pmod q,
\]
so a single public table of $2^t$ subset sums decodes every layer. The
highest-degree coefficients give an independent reverse recurrence built from
$(a_{i,n-1})_i$.

Injectivity fails only if some nonzero $d\in\{-1,0,1\}^t$ has
$\sum_ia_{i,0}d_i=0$. The vectors $d$ and $-d$ define the same event, so there
are $(3^t-1)/2$ sign classes, and in each the chosen nonzero coordinate is
$\pm1$, hence a unit modulo every $q$. The two endpoint tables depend on
disjoint independent coefficient vectors, so both fail with probability at
most $\varepsilon^2$ where $\varepsilon=(3^t-1)/(2q)$. At the published
$t=9$ and $q=18941623$,
\[
  \varepsilon=\frac{9841}{18941623},
  \qquad
  1-\varepsilon^2=0.999999730074\ldots,
\]
with tables of only $512$ entries each.

For MP-LWE the plaintext is then $\mu=c_2-\sum_ir_i\mathbin{\odot_d}b_i$. For
MP-CLWR, each inverse-rounding fiber is $1023$ or $1024$ consecutive residues,
so a public midpoint lift $L$ is at centered distance at most $512$ from the
private lift, and
\[
  \lVert v-v_0\rVert_{\infty,q}\le 9\cdot257\cdot512=1184256<q/8=2367702.875 .
\]
The reconciliation identity therefore returns the hash input, and one
random-oracle query recovers the plaintext. The IND-CPA advantage is at least
$(1-\varepsilon^2)/2=0.499999865037\ldots$

The published stopping point is not where the attack stops. The bound
$\varepsilon=(3^t-1)/2q$ measures whether an endpoint table is \emph{injective},
and at this modulus it exceeds $1$ once $t\ge16$, so beyond that it asserts
nothing. The attack still works, because a collision does not defeat the
recurrence: it forces a branch, and a wrong branch fails at the next layer,
whose residual then has no preimage. The governing quantity is the expected
number of spurious preimages per layer, $2^t/q$. Equivalently the recurrence
must fix $t$ bits per layer while each layer supplies $\log_2q$ bits, so it
collapses when $2^t$ approaches $q$.

Measured at the full published dimensions over three key generations per row,
the true masks are recovered in $3/3$ trials at $t=9,12,16,20,22$ and $24$, and
in $0/3$ at $t=25$. Mean explored nodes rise from $258$ at $t\le16$ to $2526$
at $t=24$, against a minimum of $258$; the largest branch factor observed was
$6$. The collapse point matches $\log_2q=24.13$.

\begin{table}[!ht]
\centering
\small
\begin{tabular}{@{}rrrrr@{}}
\toprule
$t$ & $2^t/q$ & published $(3^t-1)/2q$ & recovered & mean nodes \\
\midrule
9  & 0.0000 & 0.00052 & 3/3 & 258 \\
16 & 0.0035 & 1.136   & 3/3 & 258 \\
20 & 0.0554 & 92.0    & 3/3 & 273 \\
22 & 0.2214 & 828     & 3/3 & 331 \\
24 & 0.8857 & 7455    & 3/3 & 2526 \\
25 & 1.7715 & 22370   & 0/3 & capped \\
\bottomrule
\end{tabular}
\caption{Mask recovery against the number of mask polynomials, at the full
published dimensions over three key generations per row. The published bound
exceeds one from $t=16$ and so asserts nothing beyond that point, while the
attack runs to $t=24$ and collapses at $t=25$, where $2^t$ passes $q$.}
\label{tab:mprange}
\end{table}

A second limit applies to every algorithm, and it turns out to be the schemes'
own requirement in disguise. An IND-CPA adversary holds $c_1$, $c_2$, and both
candidate messages, so the observable ciphertext has $(n+k)+d=1024=k+d+n$
coefficients, exactly the quantity appearing in the two security lemmas.
Counting bits, the masks stop being determined past
$t^*=(k+d+n)\log_2q/(k+1)=96.32$, while the lemmas require
$t\ge(2\lambda+(k+d+n)\log_2q)/(k+1)=97.3201$, hence $98$. The two differ by
exactly $2\lambda/(k+1)=0.9961$, one unit of $t$.

The randomness requirement is therefore not an incidental inequality: it says
precisely that the mask entropy exceeds the ciphertext size by $2\lambda$
bits. That settles why the boundary attack cannot reach the theorem regime. It
fails there because the ciphertext stops determining the masks, not because a
table became large.

The scope is exact. Setting $\lambda=128$ in the schemes' own security lemmas
requires $t\ge98$, a factor $4.1$ over the concrete attack and only $1.017$
over $t^*$. No published parameter set uses $9<t<98$, so the extended range breaks no
further printed row. What it changes is the margin, and what explains it. The reason usually given,
the size of the $2^t$ table, is not the binding constraint; under-determination
is. Bai et
al.~\cite[Section 6.2.2]{bai2019} already treat mask recovery as an attack
target, and the forward recurrence is the MPSign coefficient-layer
attack~\cite{mpsign2020} specialised to a binary subdomain. What is added here
is the exact sign-class injectivity bound, the independent reverse endpoint,
the near-one success probability over key generation, and plaintext recovery
through both encryption interfaces.

\subsection{A lattice e-voting protocol loses receipt-freeness}
\label{sec:voting}

\finding{The printed fake-receipt simulator does not reproduce the genuine
cumulative randomness law at any permitted ternary parameter, so a coercer who
receives a transcript can read the vote and reject fabricated receipts.}

Receipt-freeness is the property that a voter cannot prove to anyone else how
they voted, even if they want to. It is what stands between a secret ballot and
a market in votes, and it is usually obtained by letting the voter fabricate a
receipt for any candidate: if a convincing fake exists, a genuine receipt
proves nothing. The construction here provides such a simulator, and the
attack is that its output is distinguishable from the real thing.

In the cast-as-intended protocol of Farzaliyev, P\"arn, Saarse, and
Willemson~\cite{fpsw2025}, the receipt carries the cumulative randomness
$r_v=r_{vd}+r_s$ and the re-randomized ciphertext $c^{*}$, and public special
decryption returns the claimed vote as $c^{*}-Ar_v$ with no secret key.
Privacy therefore rests entirely on the voter's ability to fabricate an
indistinguishable receipt for a different vote.

Each ternary coefficient satisfies $\Pr[X=\pm1]=p$ and $\Pr[X=0]=1-2p$, so a
genuine cumulative coefficient is the convolution $Y=X_1+X_2$ with
\[
  \Pr[Y=\pm2]=p^2,\qquad
  \Pr[Y=\pm1]=2p(1-2p),\qquad
  \Pr[Y=0]=1-4p+6p^2 .
\]
The printed simulator keeps the positions where the genuine value is $\pm2$
and replaces every other coefficient by a uniform element of $\{-1,0,1\}$,
giving $\Pr[Y^{*}=\pm2]=p^2$ and
$\Pr[Y^{*}=-1]=\Pr[Y^{*}=0]=\Pr[Y^{*}=1]=(1-2p^2)/3$. The one-coordinate
total variation distance is exactly
\[
  \Delta(Y,Y^{*})=\frac{2(10p^2-6p+1)}{3}
   =\frac{20}{3}\Bigl(p-\frac{3}{10}\Bigr)^{2}+\frac{1}{15}
   \;\ge\;\frac{1}{15},
\]
with equality only at $p=3/10$. The bound is uniform over the permitted range
$0<p<1/2$, so no choice of the ternary parameter repairs the simulator. This
is the sharp form of the result: the defect is structural, not a parameter
accident.

The gap is concentrated in the zero mass, $\theta_R=1-4p+6p^2$ against
$\theta_F=(1-2p^2)/3$, so counting zero coefficients in one $d$-coefficient
component suffices. At the published $d=4096$ and $p=1/3$, accepting a receipt
as genuine when the zero count is at least $1210$ gives exact binomial tails
\[
\begin{aligned}
  \Pr[D(\text{fake})=\text{genuine}]&=1.046943719579\cdot10^{-7},\\
  \Pr[D(\text{genuine})=\text{fake}]&=9.114537336335\cdot10^{-8},
\end{aligned}
\]
hence $\operatorname{Adv}(D)=0.999999804160$. For the centered-binomial
marginal $p=5/16$ the optimal cutoff $1235$ gives $0.999997734193$. One
polynomial component suffices, so the attack does not depend on resolving the
source's inconsistent use of $d$ against $3d$.

Under the componentwise reading of the simulator across all $3d=12288$
coefficients, the cutoff $3630$ lowers the total error to
$2.010188913363\cdot10^{-19}$. The one-component figure is the one to quote,
because it does not depend on that reading.

The zero count is also close to the best possible statistic, and the exact
sense in which it is best determines how the result can be attacked. The
optimal advantage of any test is the total variation distance between the two
product laws. Computing that distance exactly over the full multinomial and comparing
it with the induced law of a candidate statistic shows that the pair

\[
  \bigl(\text{zero count},\ \#\{\text{coefficients equal to }\pm2\}\bigr)
\]

is sufficient: its distance equals the optimal one at every tested $d$ and for
both $p=1/3$ and $p=5/16$. The zero count alone is strictly weaker, for
instance $0.187120914287$ against the optimal $0.195640801881$ at $d=8$ and
$p=1/3$. So the exact optimal test thresholds the zero count conditionally on
the observed $\pm2$ count, and no further statistic helps. The $\pm2$ cells
carry no information alone, since the simulator reproduces them with
probability $p^2$ exactly as the genuine law does; they matter only as the
conditioning variable.

The complete coercer computes $v_T=c^{*}-Ar$, checks that it is the demanded
candidate, and then applies the zero-count test to $r$. A genuine receipt is
accepted and reveals the actual vote; the printed fake receipt is rejected. No
lattice reduction, MLWE oracle, secret key, server interaction, or malicious
verifier behaviour is required. BGV security, MLWE hardness, and ballot secrecy
against an outsider holding no transcript are not attacked.

A repair must change the protocol or the simulation strategy, since
reparameterisation is excluded by the universal bound. Candidate directions
are a fully zero-knowledge re-randomization proof, sampling the fake
cumulative randomness from the exact convolution law subject to every publicly
visible constraint, or preventing a transferable party from holding both the
randomness and the matching ciphertext.

\subsection{A recovered permutation yields the old decryption key}
\label{sec:hollow}

\finding{The published Schur-product recovery of the hidden permutation in an
updatable encryption scheme extends, by linear algebra alone, to recovery of a
valid old decryption key and decryption of the pre-update challenge.}

The Hollow-LWE updatable public-key encryption scheme of Albrecht, Ben\v{c}ina,
and Lai~\cite{abl2025} has old key pair $A^{\mathsf T}r=u$ with
$r\in\{-1,1\}^n$. An honest update applies a signed permutation $O=DP$ and a
basis change $U$, publishing $H=DPAU$ and updating the secret to
$r^{*}=DPr$. In the scheme's IND-CR-CPA experiment the adversary receives the
old-key challenge ciphertext together with the updated public key and the
exposed updated secret key.

Battagliola, Mora, and Santini~\cite{bms2025} recover $P$, because squaring
removes the signs:
\[
  (DPv)\star(DPw)=D^2P(v\star w)=P(v\star w),\qquad D^2=I .
\]
That result is theirs, it is prior public work, and it is not claimed here.
Their analysis stops at the code-equivalence layer, leaving $D$ apparently
still hiding the key.

It does not. Write
\[
  B=PA,\qquad H=DBU,\qquad
  D=\operatorname{diag}(d_1,\ldots,d_n),\quad d_i\in\{-1,1\}.
\]
Choose $k$ row indices $I=(I_1,\ldots,I_k)$ with $B_I$ invertible and set
$V=BB_I^{-1}$ and $W=HH_I^{-1}$. Since $H_I=D_IB_IU$ we have $W=DVD_I$, so at
every nonzero $V_{ij}$,
\[
  \lambda_{ij}:=W_{ij}V_{ij}^{-1}=d_id_{I_j}\in\{-1,1\}.
\]

Form a graph on the row indices $[n]$ with an edge between $i$ and $I_j$
whenever $V_{ij}\neq0$, labelled $\lambda_{ij}$. Within each connected
component $C$, multiplying edge labels determines public signs
$\delta_i\in\{-1,1\}$ with $d_i=\varepsilon_C\delta_i$ for one unknown
$\varepsilon_C\in\{-1,1\}$ per component. Consistency around cycles is
guaranteed by the existence of the true $D$. A component containing no index
$I_j$ consists of rows whose corresponding row of $V$ is zero, and its sign
does not affect $B$.

The component signs are then fixed by public data, which is the step that makes
the recovery unconditional. The exposed updated secret satisfies
$r^{*}=DPr$, so for a choice of component signs write
$\widetilde r(\varepsilon)_i=\delta_i\varepsilon_{C(i)}r^{*}_i$, the true
choice giving $\widetilde r=Pr$. The old public-key relation
$A^{\mathsf T}r=u$ is $B^{\mathsf T}\widetilde r=u$, and since $B=VB_I$ this is
$V^{\mathsf T}\widetilde r=y$ with $y:=B_I^{-\mathsf T}u$. Put
$g_j:=\sum_{i=1}^nV_{ij}\delta_ir^{*}_i$. If $C_j$ is the component containing
$I_j$, then $V_{ij}\neq0$ forces $C(i)=C_j$, so
$\bigl(V^{\mathsf T}\widetilde r(\varepsilon)\bigr)_j=\varepsilon_{C_j}g_j$.
For every component meeting some $I_j$ with $g_j\neq0$, set
$\varepsilon_C=y_jg_j^{-1}$; the true update relation guarantees this lies in
$\{-1,1\}$ and does not depend on which nonzero coordinate of the component is
used. If all such $g_j$ vanish, the same relation forces the corresponding
$y_j$ to vanish and either sign satisfies the equation. Components meeting no
$I_j$ are free.

This produces $\widetilde r'\in\{-1,1\}^n$ with $B^{\mathsf T}\widetilde r'=u$,
hence $r':=P^{\mathsf T}\widetilde r'\in\{-1,1\}^n$ with
$A^{\mathsf T}r'=u$: a valid old decryption key. If the component signs are
uniquely determined then $r'=r$; otherwise the public data admit several
compatible signed secrets and any of them serves. For an old-key ciphertext
$c_0=Ax+e$ and
$c_1=\langle u,x\rangle+e_0+\lfloor q/p\rceil m$,
\[
  c_1-\langle r',c_0\rangle
  =e_0-\langle r',e\rangle+\Bigl\lfloor\tfrac qp\Bigr\rceil m,
\]
because $A^{\mathsf T}r'=u$. Since $r'\in\{-1,1\}^n$ carries the norm bound of
an honestly generated key, the scheme's own correctness argument returns the
challenge message.

A compatible signed update follows. With
$D'=\operatorname{diag}(\delta_i\varepsilon_{C(i)})$ and
$U'=B_I^{-1}D'_IH_I$, the edge relations give $W=D'VD'_I$ including the zero
entries, hence $H=D'BU'$. Permutation recovery therefore lifts by linear
algebra to a valid signed update relation and an old decryption key, both
publicly checkable against $H=DPAU$.

After the permutation is known the work is two $k\times k$ inversions, products
of $n\times k$ by $k\times k$ matrices, graph propagation, and the public-key
consistency equations: $O(k^3+nk^2)$ field operations with classical matrix
arithmetic. There is no sign search and no lattice reduction. If permutation
recovery succeeds with probability $\varepsilon_P$, the IND-CR-CPA advantage is
about $\varepsilon_P/2$.

Reduced-dimension runs over $\mathbb F_{3329}$ recovered the update relation,
the originally sampled old key, and the challenge plaintext in $12/12$ trials
at each of
$(n,k)=(64,12)$, $(128,16)$, and $(256,24)$. A permutation altered by a single
transposition was rejected in $12/12$ trials. These runs validate the
post-permutation algebra; they assume $P$ and do not rerun the published
full-size recovery. The argument above requires only a valid old decryption
key; recovery of the originally sampled one is what these runs happened to
observe.

A second defect is independent of that attack. The scheme's reduction
transforms ordinary LWE of dimension $k-h$ into Hollow LWE of width $k$ with
hull dimension $h$, but the appended parameter script invokes the lattice
estimator with $n=k$. The printed security labels are therefore not certified
by the source's own reduction and estimator combination. This does not by
itself place the rows below those levels; it shows the certification is
computed at the wrong source dimension. A corrected search must estimate at
$k-h$, raise $k$, recompute $n$, $q$, correctness, and $h$, and iterate, since
changing $k$ changes all of them.

Everything the published script does apart from the estimator call is
arithmetic, and reproducing it recovers the source's parameter tables exactly,
including $n$, $\log_2q$, both components of $h$, and both size columns on all
four rows. Two consequences follow, and they run against each other.

Repairing the dimension defect alone makes the code-equivalence exposure
worse. It suffices to raise $k$ until $k'-h'\ge k$, since the script already
certified dimension $k$ and hardness is monotone in the dimension, so that
direction needs no estimator call. But $n=\lceil(1+c)k\log_2q\rceil$ grows with
$k$ while $h$ does not, so every repaired row sits further inside the weak
regime, with $\sqrt{2n}$ rising from $120.9$ to $127.5$ on the first row and
similarly on the others.

Repairing both requires a fixed point, because raising $h$ above $\sqrt{2n}$
forces $k$ up to preserve $k-h\ge k_0$, which raises $n$ and $\sqrt{2n}$
again. Iterating gives:

\begin{table}[!ht]
\centering
\small
\begin{tabular}{@{}lrrrrr@{}}
\toprule
Row & $k\to k^{*}$ & $h\to h^{*}$ & $n^{*}$ & ciphertext & update token \\
\midrule
128, $p=2$  & $450\to600$   & $27\to140$ & 9750  & $1.34\times$ & $1.33\times$ \\
128, $p=16$ & $550\to750$   & $26\to179$ & 15938 & $1.54\times$ & $1.54\times$ \\
192         & $900\to1150$  & $37\to228$ & 25875 & $1.28\times$ & $1.28\times$ \\
256         & $1250\to1550$ & $48\to272$ & 36813 & $1.24\times$ & $1.24\times$ \\
\bottomrule
\end{tabular}
\caption{Cost of repairing both Hollow-LWE defects at once. Estimating at
$k-h$ and leaving the weak Schur regime pull against each other, so the rows
are the fixed point of iterating both.}
\label{tab:hollowrepair}
\end{table}

The hull dimension has to grow by a factor between $5.2$ and $6.9$, and both
ciphertexts and update tokens grow by $1.24$ to $1.54$ times. Update tokens
were already the dominant cost, so the repaired scheme carries tokens between
roughly $1$ and $4.3$ megabytes. These rows are a lower bound on the repair,
since only the estimator-free direction is used.

\subsection{A rerandomisable ciphertext answers whether its noise wrapped}
\label{sec:reaction}

\finding{Shifting one ciphertext coefficient coherently in both CRT limbs
cancels in the plaintext until the exact-noise centre lift wraps, so the
authenticated payload's accept bit is a threshold oracle on the noise, and the
noise determines the key.}

The construction of Vodenicarevic et al.~\cite{eprint20261618} recovers the
exact noise during decryption: it computes $v_t=c_{0,t}-c_{1,t}s$ and $v_{q_2}=c_{0,q_2}-c_{1,q_2}s$,
centre-lifts the $q_2$ limb to obtain $\nu$, and returns
$M=(v_t-\nu)\Delta_t^{-1}\bmod t$. Embedding the plaintext so that it vanishes
modulo $q_2$ is what leaves that limb carrying only the noise, which the source
states as the feature enabling exact recovery without rounding. Its parameters set $\Delta_t=q_2\bmod t$,
which is the join the attack turns on, because it forces
\[
  q_2\Delta_t^{-1}\equiv1\pmod t.
\]

Add the same integer $\delta$ to one coefficient of $c_0$ in both limbs, leaving
$c_1$ untouched. While $\nu_j+\delta$ stays below $q_2/2$ the centre lift does
not wrap, the two copies of $\delta$ cancel, and the plaintext is unchanged, so
the inner authenticated payload verifies. As soon as the boundary is crossed the
lift subtracts $q_2$, and the identity above moves that one plaintext
coefficient by exactly one, so verification fails. The accept bit is therefore
the monotone predicate
\[
  O_j(\delta)=\bigl[\,\delta<\lfloor q_2/2\rfloor+1-\nu_j\,\bigr],
\]
and binary search returns $\nu_j$ exactly. The $q_2$ limb carries no message, so
once the whole noise polynomial is known $c_{1,q_2}s=c_{0,q_2}-\nu$ returns the
secret wherever $c_{1,q_2}$ is a unit; the secret is small, so centre-lifting
its residue gives the integer polynomial and the other limb with it.

The source's own scalar example reproduces this. At $t=7$, $q_2=11$,
$\Delta_t=4$, $s=2$, $M=5$ and $\nu=1$, the first rejecting shift is $\delta=5$,
which returns $\nu=1$ and then $s=2$. On a negacyclic ring at $n=64$ with a
$31$-bit $q_2$ the attack recovers every noise coefficient and the secret from
the accept bit alone, at $30$ queries per coefficient. At the source's
$q_2<2^{32}$ and $n=4096$ that is at most $4096\cdot32=131072$ adaptive
reactions from a single captured ciphertext.

The mitigation the source proposes does not reach this. It recommends returning
one generic failure indication so that a decryption oracle is not exposed, but
the attack never asks why a decryption failed: accept against generic reject is
the entire signal.

\notclaimed{The IND-CPA and IND\$ statements for the raw layer. Those concern
passive indistinguishability and this adversary is adaptive and active. What the
attack contradicts is the composed claim that wrapping the malleable layer
around an authenticated payload leaves a malicious rerandomiser with denial of
service only. The attack is conditional on the acceptance being observable: a
recipient exposing no difference between acceptance and failure admits no
oracle, and a proof that every rerandomisation is a bounded, correctly formed
encryption of zero would exclude the offsets. Both conditions are stated
wherever the result is.}

\subsection{A perturbation rank is averaged where the attacker chooses}
\label{sec:minrank}

\finding{A MinRank target raised by perturbation is computed over an ambient
dimension the attacker can publicly reduce, and at an average rank the attacker
need not accept.}

The scheme of Patarin, Vacek, and Roullet~\cite{eprint2026404} hardens HFE by
adding $r$ perturbations
$Q_i=L_iL_i'$, which raises the rank a MinRank attack must reach from $d$ to
about $d+r$. Three separate steps in that accounting do not hold. All figures
below use the source's own support-minor expression
$C_{\mathrm{MR}}(N,\mu)\approx\mu(N-1)^4\binom{2\mu+1}{\mu}^2$, so the
comparison stays inside its cost model.

The Dragon variant writes its public matrices as
$\widetilde H=\bigl(\begin{smallmatrix}\alpha&\beta\\0&0\end{smallmatrix}\bigr)$
and raises the ambient dimension from $n$ to $n+m$ because the mixed block
$\beta$ adds columns. The split between signature and hash variables is public,
so taking the principal $aa$ block of every public matrix deletes $\beta$ while
keeping the transformed $\alpha$. The target is unchanged and the dimension is
$n$: the $128$-bit row moves from $2^{131.89}$ to $2^{125.53}$, and the $80$ and
$100$-bit rows likewise fall below their labels.

For the degree-3 variant the source polarises a cubic perturbation, finds a
differential rank contribution of $0$, $1$ or $2$ with probabilities
$\tfrac18,\tfrac68,\tfrac18$, and averages it to one per perturbation. An
attacker picks the direction, so what matters is the lower tail of
$S=X_1+\cdots+X_r$. Searching for a direction with $S\le s$ and running MinRank
at $\mu=d+s$ costs $C_{\mathrm{MR}}(n+m,d+s)/\Pr[S\le s]$, and optimising $s$
puts the $128$-bit rows at $2^{100.75}$, $2^{102.22}$ and $2^{103.65}$ for
$d=2,3,4$, and the $256$-bit row at $2^{187.45}$.

The signing distribution used for the concrete estimates sets exactly one $Q_i$
to one. Every accepted signature then satisfies the degree-2 relation
\[
  1+\sum_{i=1}^rL_i(x)L_i'(x)=0,
\]
where the source's own discussion of transcript structure reaches for degree-4
relations such as $Q_iQ_j=0$. The polar form
$B_R=\sum_i(u_iv_i^{\mathsf T}+v_iu_i^{\mathsf T})$ has kernel
$\bigcap_i\ker L_i\cap\ker L_i'$ when the forms are independent, so interpolating
that one relation returns the common annihilator of the entire perturbation.
Restricted to that kernel every $LL'$ term vanishes identically and the target
returns from $d+r$ to $d$: at the $128$-bit rows the restored low-rank stage
costs about $2^{38.5}$ and $2^{35.0}$. Interpolation needs enough transcripts to
span the degree-$\le2$ monomials, $32897$ at $N=256$ and $115441$ at $N=480$.

\notclaimed{Any executed key recovery. These are statements inside the source's
own cost model and about its own signing distribution; no MinRank attack was
run at these parameters, and the restored low-rank figures are the cost of that
stage rather than of a completed recovery. The open bridge is to feed the
restricted instance through an HFE MinRank implementation. A signer that mixes
odd and even hypothesis weights destroys the constant relation, so the
exact-one tables and the randomised-weight proposal are not the same scheme.}

\subsection{An explicit normal basis splits along the CRT factors}

\finding{The constructive normal-basis generator of Proposition 6.11 has a
Galois orbit supported factorwise, so the stated degree-63 instance decomposes
into seven degree-nine instances at an unchanged noise rate.}

Bombar, Couvreur, and Debris-Alazard~\cite{bcd2022} build a normal-basis
generator inside the CRT product as $b=(a,0,\ldots,0)$ and take its complete
Galois orbit. Because the generator is supported on one component, so is every
element of the orbit. The global basis is therefore a disjoint union of local
normal bases instead of a genuinely global one.

The consequence for the noise is exact. Error sampled as
$e=\sum_{\sigma\in G}e_\sigma\sigma(b)$ with independent Bernoulli
coefficients projects onto a component by annihilating every basis element
supported elsewhere. The surviving error is a local Bernoulli error with the
\emph{same} parameter: there is no compensating increase. Measured over
$100000$ trials at source rate $1/4$, the seven per-factor rates all lay
within $0.0006$ of the source rate.

For the stated example, exhaustive trial division over $\mathbb F_2$ gives
\[
  X^{63}+X^{7}+1=\prod_{i=1}^{7}f_i(X),\qquad \deg f_i=9,
\]
with seven distinct irreducible factors and residual cofactor $1$. Hence, for
this basis,
\[
  \mathrm{NormalRingLPN}_{63}\longrightarrow 7\times\mathrm{RingLPN}_{9}.
\]
Each local secret has only $2^9=512$ candidates, so the complete CRT secret
follows from seven independent maximum-likelihood searches using
$7\cdot512=3584$ candidate tests per scan, against $2^{63}$ undecomposed.

Scope requires care here. The authors themselves warn that their constructive
basis can cancel noise under projection and suggest that the normal basis need
not be the one they construct. That warning is prior public information and is
not claimed here. What is added is the exact basis-support decomposition, the
observation that the projected rate is unchanged and not merely reduced,
the precise $63\to7\times9$ collapse, and the resulting candidate count. An
independently sampled random normal basis is untouched, and the
search-to-decision theorem is not contradicted: it can remain true while both
search and decision are easy under this particular basis.

\subsection{Semidirect group-ring LWE has two public maps}
\label{sec:group}

\finding{The augmentation defeats the full-ring decision claim, and explicit
rank-four maps defeat the selected quotient rings under their coefficient
Gaussian completion; the two maps are complementary, since the augmentation
does not descend to either quotient.}

Liu and Fu~\cite{lf2026} study LWE over full and quotient group rings for two
semidirect families, selecting quotients intended to remove the
one-dimensional representations that enable the simplest homomorphism attacks.

Their decision statement is ill-typed before any attack: the definition asks
for a probability distribution over error distributions, while the theorem
supplies only the set $\Psi_{\le\alpha}$ of elliptical Gaussians. The analysis
below fixes an explicit coefficient-basis Gaussian completion and attributes
no formal quotient statement to the source.

\paragraph{The full rings.}
Every integral group ring carries the augmentation
\[
  \varepsilon\Bigl(\sum_{g\in G}c_gg\Bigr)=\sum_{g\in G}c_g,
\]
a surjective ring homomorphism onto $\mathbb Z$. A genuine sample
$b=sa/q+e \bmod R$ projects to the one-dimensional sample $B=SA+E \bmod q$
with $S=\varepsilon(s)$, $A=\varepsilon(a)$, and $E=q\varepsilon(e)$, while a
uniform sample stays uniform. Since $\varepsilon$ sums $N=|G|$ independent
coefficients, $E$ has width at most $\alpha q\sqrt N$.

Take the sequence permitted by the source's own numerical condition:
\[
  N^{10}<q_N<2N^{10},
  \qquad \alpha_N=\frac{2N}{q_N},
  \qquad B_N=4N^{3/2}\log N,
\]
so that $\alpha_Nq_N=2N$. Given two samples, enumerate $S'\in\mathbb Z_{q_N}$
and accept when both residuals are at most $B_N$ in absolute value. The
genuine secret passes by Gaussian concentration, since the error width is at
most $2N^{3/2}$ and the ratio of $B_N/2$ to it is $\log N$. For the uniform
distribution, each candidate passes with probability at most
$(2B_N/q_N)^2$, so a union bound over $q_N$ candidates gives
\[
  q_N\Bigl(\frac{2B_N}{q_N}\Bigr)^{2}=\frac{4B_N^2}{q_N}
  =64\,N^{-7}\log^{2}N .
\]
Enumeration costs $q_N=N^{10+o(1)}$, which is polynomial in the encoded
instance size. The natural universal repair of the full-ring decision claim is
therefore false: there is an explicitly permitted coefficient-Gaussian
sequence for which the problem is polynomial-time distinguishable.

At the stated parameter sequence, the two-sample test accepted the genuine
distribution $40/40$ and the uniform distribution $0/40$ for Type I with $m=4$
and $m=6$ and Type II with $k=1$ and $k=2$. The one-sample control accepted
uniform $40/40$, confirming that the second sample is load-bearing.

\paragraph{The quotients.}
The augmentation does not reach them. Exact computation gives
$\varepsilon(N_k)=3^k$ for the Type-II norm quotient and
$\varepsilon(t^2+1)=2$ for Type I, so $\varepsilon$ survives only modulo $3^k$
and modulo $2$ respectively. The selected quotients require separate maps, and
they exist: for Type I with $n=4$, reduction of the exponent of $s$ modulo two
gives a rank-four quotient; for Type II, coordinate reduction modulo three
descends to a ring epimorphism $R_k\twoheadrightarrow R_1$ because
$N_k\mapsto3^{k-1}N_1$. Both targets satisfy $R_1/qR_1\cong M_2(\mathbb F_q)$
for the attack primes. Enumerating the rank-four quotient secret and testing
two projected residuals distinguishes in $N^{40+o(1)}$ operations with uniform
false acceptance $N^{-28+o(1)}$.

For Type I the scope of this quotient mechanism can be made exact. Since
$t^2=-1$ gives $t^{-1}=-t$, the defining relation $sts^{-1}=t^{-1}$ becomes
$st=-ts$; write $R^I_{m,q}$ for the resulting rank-$2m$ algebra modulo $q$, and
take $m$ even, $q$ odd, and $q\equiv1\pmod m$, which Dirichlet permits inside
the parameter sequence. For each $m$-th root of unity $z\in\mathbb F_q^\times$,
\[
  \rho_z(s)=\begin{pmatrix}z&0\\0&-z\end{pmatrix},
  \qquad
  \rho_z(t)=\begin{pmatrix}0&-1\\1&0\end{pmatrix}
\]
satisfy $\rho_z(s)^m=I$, $\rho_z(t)^2=-I$, and
$\rho_z(s)\rho_z(t)=-\rho_z(t)\rho_z(s)$, so $\rho_z$ is a homomorphism
$R^I_{m,q}\to M_2(\mathbb F_q)$. The representations for $z$ and $-z$ are
conjugate, and taking one representative from each pair $\{z,-z\}$ gives
\[
  R^I_{m,q}\;\cong\!\!\prod_{z\in\mu_m(\mathbb F_q)/\{\pm1\}}\!\! M_2(\mathbb F_q).
\]
Injectivity is direct: every element is uniquely $x=A(s)+B(s)t$ with
$\deg A,\deg B<m$, and
\[
  \rho_z(x)=\begin{pmatrix}A(z)&-B(z)\\B(-z)&A(-z)\end{pmatrix},
\]
so if all blocks vanish then $A$ and $B$ vanish at all $m$ roots of $X^m-1$ and
are zero. Source and target both have dimension $2m$ over $\mathbb F_q$, so the
map is an isomorphism. Verified exactly for $m=2,4,6,8,10,12,16,20$.

That decomposition turns the recorded limitation into a bound. For
$x=\sum_{i<m}(a_is^i+b_is^it)$ write
$\|x\|_{\mathrm{coef},\infty}=\max_i\{|a_i|,|b_i|\}$, and let
$\|\cdot\|_{\infty,q}$ be the largest absolute centred representative of a
matrix entry. For the block $z=1$ every entry is a signed sum of $m$
coefficients, so
\[
  \|\rho_1(x)\|_{\infty,q}\le m\|x\|_{\mathrm{coef},\infty}.
\]
Every other block distorts far more. Let $z\notin\{\pm1\}$ have
multiplicative order $d$. Then
\[
  \max_{0\le j<d}\bigl|[z^j]_q\bigr|\;\ge\;\sqrt{\tfrac q2},
\]
where $[\cdot]_q$ is the centred representative. Otherwise every element of
$\langle z\rangle$ would have centred representative below $\sqrt{q/2}$ in
absolute value, the product of any two would stay below $q/2$, and
multiplication in the subgroup would agree with integer multiplication without
wraparound. The centred representatives would then form a finite multiplicative
subgroup of the nonzero integers, which is contained in $\{\pm1\}$, contradicting
$z\notin\{\pm1\}$. Choosing $j$ attaining the bound, the basis element $s^j$ has
$\|s^j\|_{\mathrm{coef},\infty}=1$ while
$\rho_z(s^j)=\operatorname{diag}(z^j,(-z)^j)$ has
$\|\rho_z(s^j)\|_{\infty,q}\ge\sqrt{q/2}$.

Because the split algebra is a direct product of simple matrix rings, every
two-sided quotient retains a subset of the blocks. So when $m<\sqrt{q/2}$, the
rank-four $z=1$ image is maximal among quotient images whose
coefficient-to-centred distortion stays below $\sqrt{q/2}$: any strictly larger
quotient contains a nontrivial block and loses that property. This is the
maximality the bounded-residual attack needs, and it is not the claim that
every statistical attack on the other blocks is impossible. Measured at $m=20$
and $q=100361$ over $400$ errors with coefficients in $\{-1,0,1\}$, the largest
centred entry was $12$ for the $z=1$ block against $|G|=80$, and between $48253$
and $50167$ for each of the nine others against $q/2=50180$, consistent with the
two bounds. A repair must remove that single trivial-character direction.

Both results recover only the image of the secret under the respective map,
and neither instantiates an attack on a fully specified encryption scheme. The
source supplies no Type-II parameter generator or message encoding. Standard
RLWE is not affected. Removing one-dimensional representations is, on this
evidence, insufficient: what must be excluded is every efficiently computable
image in which the secret can be enumerated and the projected error remains
distinguishable.

\subsection{An unrestricted hint definition is not generically hard}

\finding{Read without the restrictions its hardness theorem imposes, the
Hint-MLWE definition admits a parameter choice that reveals the randomness
exactly.}

The parameterized definition of Kim, Lee, Seo, and Song~\cite{klss2023} admits
$\ell=1$, $\gamma=1$, and a point-mass hint distribution. That choice makes the
hint reveal the randomness exactly and gives a distinguisher of advantage
$1-q^{-nm}$.

This does \emph{not} contradict the source's Gaussian hardness theorem, which
excludes that parameter choice. It bounds how the assumption may be cited: a
construction may invoke only the restricted family the theorem covers, not
``Hint-MLWE'' generically.

\section{Where a proof covers less than the construction needs}
\label{sec:scope}

The three targets in this section behave differently from everything above.
Nothing in them is wrong, and nothing in them is attacked.

In each case the pattern is the same. A paper proves that some problem is hard
under a stated set of conditions. A construction is then built on the same
assumption, but relies on it holding in a wider range than the proof covers.
Nobody closes that gap, and usually nobody claims to: the authors are often
explicit about what their theorem does and does not reach. The difficulty
arrives later, when the assumption is cited by name, because the name carries
the wide reading while the proof only supports the narrow one.

So what these three entries support is a statement about how an assumption may
be cited, not a statement about whether a scheme is secure. Presenting them
alongside the findings above would make this record claim more than its
evidence does, and calling them breaks would be simply false. They are here so
that the next reader does not have to rediscover the difference, and separate
so that nobody mistakes the difference for an attack.

\paragraph{Sparse LWE.}
Jain, Lin, and Saha~\cite{jls2024} derive hardness from ordinary LWE in
dimension about the sparsity $k$, so the reduction reads
$\mathrm{LWE}_k\lesssim\mathrm{SparseLWE}_{n,k}$ and not
$\mathrm{LWE}_n\lesssim\mathrm{SparseLWE}_{n,k}$. Raising $n$ at fixed $k$
does not inherit the hardness normally associated with dimension $n$, and a
blockwise search at constant $k$ is polynomial. The authors explicitly study
constant and polylogarithmic sparsity and present this picture themselves, so
this is not recorded as a finding against them. The defensible statement is
about citation, not correctness: security is governed by $k$.

\paragraph{Tensor LWE.}
Agrawal, Rossi, Yadav, and Yamada~\cite{aryy2023} reduce standard LWE to
Tensor LWE for $x_i=0$ and for one common fixed $x$, and explain why the
argument does not extend to several distinct $x_i$ sharing one secret. The ABE
construction nevertheless invokes the full varying-$x_i$ assumption, and its
theorem is additionally very selective. A cancellation direction using public
$\lambda_i$ with $\sum_i\lambda_ir_i=0$ and
$\sum_i\lambda_ir_i[j]x_i[h]=0$ would cancel the secret-dependent terms, but
requires a relation short enough that $\lvert\sum_i\lambda_ie_i\rvert\ll q$,
which is itself an SIS-type problem. No complexity bridge was found.

\paragraph{Succinct LWE.}
Wee~\cite{wee2024} gives an ordinary-LWE argument for the wide regime
$\hat m=\ell m$, while the construction that provides the headline compression
uses $\hat m=m$, supported by adding public-coin evasive LWE. The succinctness
gain lives precisely in the part of the assumption that ordinary LWE does not
establish. Turning the public trapdoor into a short dual relation requires a
short vector in $\ker T_W\cap\ker(I_\ell\otimes G)$, again an SIS problem, and
the source itself discusses the same danger in the extreme narrow regime.

None of the three is a reason to distrust the paper it comes from. Each is a
reason to check, before relying on one of these assumptions, which version of
it the cited theorem actually establishes.

\section{Reproducibility}

Every check below uses only the Python standard library, runs in seconds, and
contains at least one discriminating control. A control fixes its expected
outcome before running and shows that the apparatus separates the hypothesised
phenomenon from a neighbouring case; whether it is required to pass or to fail
is secondary. Table~\ref{tab:repro} lists what each check establishes and what
its control is.

\begin{table}[!t]
\centering
\footnotesize
\begin{tabular}{@{}L{3.5cm}L{4.5cm}L{4.2cm}@{}}
\toprule
Check & Establishes & Control, with its required outcome \\
\midrule
Ring-LWR opening &
Universal double opening, the witness-free prover, the false-statement
density, and the nonunit challenge difference &
Honest opening with multiplier one still verifies \\
\addlinespace
Spinel encoding &
The two collisions, the one-byte second preimage, the $244/255$ count, the
exhaustive two-byte enumeration, and the $\mathrm{SL}_4$ digest space &
The fixed-width six-trit repair must separate every colliding pair, and the
per-byte codeword map must be injective \\
\addlinespace
Ring-LWR extraction density &
The unit criterion, and that extraction is undefined for
$\tfrac12-1/(2^{257}-2)$ of challenge pairs &
Brute-force inversion on two small rings must agree with the parity
criterion \\
\addlinespace
CRT-RLWE reaction &
The source's scalar example, and full noise plus key recovery on a ring at
$n=64$ from the accept bit alone &
Honest decryption must reproduce the plaintext, and a single-limb shift must
never preserve it \\
\addlinespace
MinRank accounting &
The projected dimension, the optimised lower tail, and the restricted
low-rank stage, all under the source's own expression &
The per-perturbation rank law must reproduce the source's average of one \\
\addlinespace
Middle-product boundary attack &
All $2313$ masks and both plaintexts at the full published dimensions, and
the midpoint-lift bound &
The reconciliation bound is checked, not assumed \\
\addlinespace
Middle-product range &
Recovery at $t\le24$ and collapse at $t=25$, three key generations per row &
$t=25$ must exhaust the node budget \\
\addlinespace
Receipt simulator &
The exact distance, the $1/15$ bound, the binomial tails, and sufficiency of
(zero count, $\#\pm2$) &
A simulator drawn from the correct convolution law must show no advantage \\
\addlinespace
Hollow-LWE bridge &
A valid old key recovered, and the challenge decrypted, $12/12$ at three
dimensions &
A permutation altered by one transposition must be rejected \\
\addlinespace
Hollow-LWE parameter repair &
Reproduction of the source's own two tables, and the cost of repairing both
defects &
The reproduction must match every published row before anything is built
on it \\
\addlinespace
Constructive basis collapse &
The exact seven-way degree-nine splitting and the unchanged projected noise
rate &
A perturbed factor must break the product, and a neighbouring trinomial must
not split the same way \\
\addlinespace
Group-ring augmentation &
Multiplicativity, the union bound, separation at the stated parameters, and
the count of sign characters &
One sample instead of two must fail to separate \\
\addlinespace
Type-I block maximality &
The complete splitting, the $m$ distortion bound for the trivial-character
block, and the $\sqrt{q/2}$ lower bound for every nontrivial block &
The $z=1$ block must stay below the $m$ bound; every nontrivial block must
contain a basis witness of centred norm at least $\sqrt{q/2}$ \\
\bottomrule
\end{tabular}
\caption{Each executable check, what it establishes, and its control.}
\label{tab:repro}
\end{table}

Two limitations apply. The Hollow-LWE runs assume the
permutation and validate only the post-permutation algebra at reduced
dimensions; they are not a benchmark of the full-size recovery, which is other
authors' work. The corrected Hollow-LWE parameter rows are a lower bound on
the repair, because only the estimator-free direction is used.

\deleted{Source provenance was checked against the primary documents. The
Hollow-LWE row values, the split of the hull dimension, and the verbatim
estimator call are confirmed against that paper's Table 1, Table 2, and
appended script; the weak-regime condition and the reported recovery are
confirmed against the abstract of the code-equivalence paper. Novelty
statements throughout are bounded by the prior-art passes performed.}\changegap
\added{For Hollow-LWE, Tables 1 and 2 of~\cite{abl2025} supply the row values
and hull split, and the appended script fixes the estimator call. The
permutation-recovery premise comes from~\cite{bms2025};
Section~\ref{sec:hollow} supplies the bridge to old-key and plaintext
recovery.}
\deleted{The Spinel source is pinned to arXiv version 2, and the deterministic
checker reproduces the encoding witnesses and the fixed-width repair control.
No implementation-level forgery or disclosure record is claimed here.}

\section{What a security argument should enumerate}

\deleted{The case studies above demonstrate a single methodological claim:
autonomous discovery can reach reproducible candidate insights before human
review, provided that promotion remains evidence-gated.}\changegap
\added{Together, the case studies show that autonomous search can discover
reproducible failures at representation and distribution boundaries before
human review.}

\deleted{None of the defects above required a new lattice technique.}
\deleted{None of the defects above required a new attack on the advertised hard
problem.}\changegap\added{Each defect arises before the advertised hard problem
becomes relevant.} They required
reading what each construction publishes, asking which efficiently computable
maps survive, and checking whether each asserted distribution is the one that
actually occurs. \deleted{Machine assistance makes it cheap to ask those two
questions across many papers at once; it contributes nothing to answering
them.} \deleted{Machine assistance makes it cheap to pose those two questions
across many papers and can autonomously propose and execute candidate answers.
The evidence gate admits reproducible insights; subsequent human review
confirms their source fidelity, consequence, scope, and provenance.} Every
result here survived because a short deterministic program with a control
confirmed it, and several plausible candidates did not survive that step.

The recommendation is narrow. When a construction introduces a new
algebraic setting \added{or representation layer}, the security argument should enumerate the efficiently
computable images of its published objects and state, for each, why the secret
is not enumerable there and why the projected error is not distinguishable.
\added{An input encoding must additionally be shown injective or uniquely
decodable before security properties are transferred through it.}
Removing one class of maps, as the semidirect quotients remove one-dimensional
representations, is not sufficient. Separately, when a verifier accepts a set
of values, binding and soundness must be proved for that exact set and not
for the honest element of it, and when a parameter script certifies a security
level, it must be invoked at the dimension the reduction actually supplies.

\bmhead{AI disclaimer}

In this paper about AI use, AI tools were used

\appendix

\section{\texorpdfstring{\snail[1.4]\;}{}The grind calculus}
\label{app:calculus}

Section~\ref{sec:grind} states the method in prose. The calculus below states
it in notation, fixing which discipline each step enforces and where a
proposal becomes evidence. \deleted{Admission produces a reproducible candidate
insight; it does not by itself produce a cryptanalytic claim. Promotion to a
break, defect, gap, or attribution additionally requires the human audit
defined below.}\changegap\added{Admission assigns the intermediate judgement
$\mathsf{ev}$; the classification rules below assign a final result type after
human audit.} The two axioms carry the weight.

\subsection{The arena}

\begin{definition}[Arena]
Let $\Pi$ be a set of source documents. Each $\pi\in\Pi$ carries a
\emph{version stamp} $\nu(\pi)$, meaning a revision date together with an
archive identifier. The predicate $\mathrm{Ver}(\pi)$ holds when
\deleted{the revision}\changegap\added{the source version}
under attack is recorded, which is needed only because a preprint may be revised
after the attack and a reader needs to know which text was read. Write
\[
  \mathcal J(\pi)\;=\;\{\,\text{joins of }\pi\,\},
\]
the finite set of points at which the argument of $\pi$ crosses a boundary: a
change of algebra, \added{a representation or encoding change,} a distribution carried through a map, a rounding or CRT
conversion, a quotient, an exceptional value, a quantifier widened from one
object to many, or a parameter substituted into a script instead of into its
theorem. The \emph{hypothesis space} $\mathcal H$ consists of triples
\[
  \eta=(\gamma_\eta,\;\rho_\eta,\;\omega_\eta),
\]
a claim $\gamma_\eta$ of $\pi$, a mathematical or state relation $\rho_\eta$
alleged to violate it, and a predicted observable $\omega_\eta$ fixed
\emph{before} anything is run.
\end{definition}

\begin{definition}[Search directions]
$\mathfrak D$ is not turned loose on $\mathcal J(\pi)$ unaimed. At each join
$j$ it is put to two questions, written $\mathsf A$ and $\mathsf B$:
\[
\begin{aligned}
  \mathsf A(j)&:\ \text{which efficiently computable }\phi\text{ survives }j
                  \text{ with useful signal},\\
  \mathsf B(j)&:\ \text{which law actually results from }j.
\end{aligned}
\]
These generate $\mathcal H_{\mathsf A}(j)$ and $\mathcal H_{\mathsf B}(j)$, and
$\mathcal H=\bigcup_{j}\bigl(\mathcal H_{\mathsf A}(j)\cup
\mathcal H_{\mathsf B}(j)\bigr)$. Section~\ref{sec:axes} states the same two
questions in prose. A hypothesis in $\mathcal H_{\mathsf A}(j)\cap
\mathcal H_{\mathsf B}(j)$ answers both at one join, and those are the ones
that reach the highest levels of $\Lambda$.
\end{definition}

\subsection{The agents}

Five agents act on the arena. They are written in fraktur because they are not
functions in any respectable sense; $\mathfrak D$ in particular is a sampler
whose distribution nobody can write down, and which is steerable only through
the join it is given and the two questions asked there.

\begin{table}[!ht]
\centering
\small
\begin{tabular}{@{}llL{7.1cm}@{}}
\toprule
Agent & Type & Duty \\
\midrule
$\mathfrak D$ \emph{(the Dreamer)} & $\mathcal J(\pi)\to 2^{\mathcal H}$ &
Proposes many $\eta$ per join. Optimistic, fluent, usually wrong \\
$\mathfrak X$ \emph{(the Executor)} & $\mathcal H\to\{\top,\bot,\mathord{?}\}$ &
Runs the exact check and its control. Returns $\mathord{?}$ for timeout,
cancellation, or tool failure \\
$\mathfrak E$ \emph{(the Executioner)} & $\mathcal H\to\{\bot\}$ &
Retires $\eta$ with the discriminating evidence attached. Never deletes \\
$\mathfrak P$ \emph{(the Pusher)} & $\mathcal H\to\mathcal H$ &
Sends a survivor with limitation $\ell$ to $\partial_\ell\eta$, the hypothesis
that $\ell$ is removable \\
$\mathfrak S$ \emph{(the Scribe)} & $\mathcal H\to\Xi$ &
Commits verdicts to the evidence store $\Xi$, including the rejected ones \\
\bottomrule
\end{tabular}
\caption{The five agents of the autonomous loop, their types, and what each is
responsible for. Only $\mathfrak X$ decides anything; the rest propose, retire,
re-aim, or record.}
\label{tab:agents}
\end{table}

\deleted{These five agents form the autonomous panel. A human operator stands
outside it and controls promotion. The operator audits a candidate handoff
for source fidelity, cryptanalytic consequence, scope, provenance, and
disclosure before any final result type is assigned.}\changegap
\added{The five agents execute the autonomous loop. Final classification is
performed by a human auditor through the predicate $\mathrm{Aud}$ defined
below.}

\begin{axiom}[Confident Nonsense]
\label{ax:nonsense}
Let $\kappa(\eta)$ be the confidence $\mathfrak D$ attaches to $\eta$, however
expressed. Then $\kappa$ is \emph{inadmissible}: no rule below may mention it,
and no ranking, ordering, or triage used to reach a judgement may be derived
from it.
\end{axiom}

Axiom~\ref{ax:nonsense} is the whole reason the calculus exists, and it is
stated as an exclusion on purpose. Calibrating $\kappa$ against the truth would
require an empirical claim about a particular generator at a particular time,
which would date immediately and would have to be redone for the next one.
Excluding $\kappa$ requires nothing of the generator at all. A proposal that is
wrong at random is harmless; a proposal that is wrong \emph{fluently} will
otherwise supply its own peer review.

\deleted{Nothing forbids $\kappa$ outside the rules. Ordering a queue by it is
free, because the admission boundary is unchanged by the order in which
candidates arrive at it.}
\deleted{Generator confidence may affect scheduling, but never evidentiary
admission or classification. Under a finite budget, the scheduling rule and
every candidate not reached within the budget must be recorded, because
scheduling changes coverage even though it cannot change the type of admitted
evidence.}\changegap
\added{Generator confidence may order the queue but does not enter admission
or classification. Under a finite budget, record the scheduling rule and every
generated candidate left unevaluated, since scheduling determines coverage.}

\subsection{Predicates}

The rules quote the conditions in Table~\ref{tab:predicates}. Each is a
property of a single hypothesis $\eta\in\mathcal H$, except $\mathrm{Brg}$,
which relates a hypothesis to a security property $P$. They divide into three
jobs. $\mathrm{Ex}$, $\mathrm{Wit}$, $\mathrm{Ctl}$ and $\mathrm{Par}$
decide whether $\eta$ becomes evidence at all. $\mathrm{Prm}$ and
$\mathrm{Brg}$ decide whether evidence reaches a security property.
$\mathrm{Clm}$, $\mathrm{Nar}$ and $\mathrm{Att}$ decide which type it
receives, and $\mathrm{Aud}$ is the human step that every classification rule
requires.

Two symbols need separating, because they look alike and do opposite jobs.
Throughout, $\eta$ is the hypothesis under test, and $\bar\eta$ is its
\emph{control}: a second hypothesis, built to differ from $\eta$ in the one
respect the analysis claims is doing the work, and run through the same
apparatus. The pair exists so that the apparatus can be shown to tell them
apart.

\begin{table}[!ht]
\centering
\small
\begin{tabular}{@{}lL{9.9cm}@{}}
\toprule
Predicate & Holds when \\
\midrule
$\mathrm{Ex}(\eta)$ &
Every decision on the path to $\omega_\eta$ is taken in exact rational or
integer arithmetic. Floating point may inform, never decide \\
$\mathrm{Wit}(\eta)$ &
A minimal deterministic witness reproduces $\omega_\eta$ \\
$\mathrm{Ctl}(\eta)$ &
A control $\bar\eta$ exists with its outcome $\mathfrak X(\bar\eta)$ fixed
before it runs, separating $\eta$ from a neighbouring case; the check is void
if $\bar\eta$ returns anything else \\
$\mathrm{Par}(\eta)$ &
The witness was run at the parameter sequence the analysis specifies, not at a
convenient one \\
$\mathrm{Prm}(\eta)$ &
The trigger uses only values the target itself permits \\
$\mathrm{Brg}(\eta,P)$ &
A validated bridge carries $\rho_\eta$ to an observable violation of $P$ \\
$\mathrm{Clm}(\eta)$ &
$\rho_\eta$ falsifies a claim the source states, or leaves it unsupported \\
$\mathrm{Nar}(\eta)$ &
What the source proves is materially narrower than what a consumer of it
relies on, with no stated claim falsified \\
$\mathrm{Att}(\eta)$ &
The relation is due to other authors \\
\added{$\mathrm{Aud}(\eta)$} &
\deleted{A human operator has audited a handoff bundle containing the exact
claim and source version, witness, precommitted control, code, execution
environment, run record, limitations, and preliminary result type, and has
checked source fidelity, consequence, scope, provenance, and disclosure}
\changegap\added{The handoff contains the exact claim and source version,
witness, control, code, environment, run record, limitations, and proposed
type, and a human auditor has verified the source, consequence, scope, and
attribution} \\
\bottomrule
\end{tabular}
\caption{The predicates the rules quote. $\eta$ is the hypothesis under test,
$\bar\eta$ its control, $\omega_\eta$ its predicted observable, $\rho_\eta$
the relation it alleges, and $P$ a security property.}
\label{tab:predicates}
\end{table}

$\mathrm{Ctl}$ is a demand for \emph{discrimination}, not for a particular
verdict. A negative control that must fail shows the apparatus can register a
negative; a positive control that must pass shows it reproduces a case already
known to hold. Both appear in Table~\ref{tab:repro}. What is excluded is a
check whose outcome would have been accepted either way.

\subsection{Evidence types, and the axiom that keeps them apart}

\begin{definition}[Evidence types]
The type of a survivor is drawn from
\[
  \mathbb T=\{\;\mathsf{brk}(P),\;\mathsf{def},\;\mathsf{gap},\;
  \mathsf{kno}\;\},
\]
a break of property $P$, an assumption or certification defect, a coverage
gap, and a known correction attributable to other authors.
\end{definition}

\begin{axiom}[No Coercion]
\label{ax:nocoerce}
Regard $\mathbb T$ as a category whose objects are the four types. Its only
morphisms are the identities. In particular there is no arrow
$\mathsf{def}\to\mathsf{brk}$, no arrow $\mathsf{gap}\to\mathsf{def}$, and no
composite producing one. A promotion may be obtained only by deriving the
target judgement from scratch through Rule~\ref{r:brk}.
\end{axiom}

\begin{grindrule}[Counting]
Let $\#$ be the tally measure on $\Xi$. Then $\#\,\mathsf{kno}=0$: results
attributable to others are cited at the point of use and contribute nothing to
the count, however load-bearing they are. The Schur-product permutation
recovery that Section~\ref{sec:hollow} builds on is the case that matters here.
\end{grindrule}

\subsection{Admission and promotion}

\begin{grindrule}[Admission]
\[
  \frac{\mathrm{Ver}(\pi)\quad \mathrm{Ex}(\eta)\quad \mathrm{Wit}(\eta)
        \quad \mathrm{Ctl}(\eta)\quad \mathrm{Par}(\eta)}
       {\Xi\vdash\eta:\mathsf{ev}}\;(\textsc{Adm})
\]
\deleted{Nothing is evidence before this line. A proposal of $\mathfrak D$,
however detailed, is at most a request for a check.}
\deleted{Nothing is computational evidence before this line. Its conclusion is
a reproducible candidate insight, not a break or other cryptanalytic claim. A
proposal of $\mathfrak D$, however detailed, is at most a request for a check.}
\changegap\added{The conclusion $\Xi\vdash\eta:\mathsf{ev}$ records a
reproduced candidate; Rule~\ref{r:brk} assigns its cryptanalytic type.}
\end{grindrule}

\begin{grindrule}[Classification]
\label{r:brk}
\[
  \frac{\added{\mathrm{Aud}(\eta)\quad}\mathrm{Att}(\eta)}
       {\Xi\vdash\eta:\mathsf{kno}}\;(\textsc{Att})
  \qquad
  \frac{\Xi\vdash\eta:\mathsf{ev}\quad \neg\,\mathrm{Att}(\eta)
        \quad \added{\mathrm{Aud}(\eta)}
        \quad \mathrm{Prm}(\eta)\quad \mathrm{Brg}(\eta,P)}
       {\Xi\vdash\eta:\mathsf{brk}(P)}\;(\textsc{Brk})
\]
\[
  \frac{\Xi\vdash\eta:\mathsf{ev}\quad \mathrm{Clm}(\eta)
        \quad \added{\mathrm{Aud}(\eta)}
        \quad \neg\,\mathrm{Brg}(\eta,P)\ \text{for every }P}
       {\Xi\vdash\eta:\mathsf{def}}\;(\textsc{Def})
\]
\[
  \frac{\Xi\vdash\eta:\mathsf{ev}\quad \mathrm{Nar}(\eta)
        \quad \added{\mathrm{Aud}(\eta)}
        \quad \neg\,\mathrm{Clm}(\eta)}
       {\Xi\vdash\eta:\mathsf{gap}}\;(\textsc{Gap})
\]
\end{grindrule}

Absence of a bridge does not on its own make a defect, which is why
$\textsc{Def}$ carries $\mathrm{Clm}$. A validated relation that falsifies
nothing the source states and narrows nothing a consumer relies on receives no
type at all: it is an observation, it stays in $\Xi$, and it is not a finding.
$\textsc{Att}$ takes precedence over the other three, so a result belonging to
other authors cannot be relabelled by finding a bridge for it.

\deleted{The preliminary type carried by the autonomous handoff is advisory and
has no force in these rules. In particular, it cannot replace
$\mathrm{Aud}(\eta)$ or coerce the type selected after audit.}\changegap
\added{The proposed type in the handoff is not a premise of
Rule~\ref{r:brk}.}

The four rules are deliberately not exhaustive. Every type on the left of the
turnstile has a premise that must be established by argument, and a hypothesis
meeting none of them is left unclassified rather than pushed into the nearest
available box.

\begin{grindrule}[Unresolved]
\label{r:unr}
\[
  \frac{\mathfrak X(\eta)=\mathord{?}}
       {\Xi\vdash\eta:\text{unresolved}}\;(\textsc{Unr})
\]
A timeout, an out-of-memory kill, a cancellation, or a tool failure is an
unresolved evidence path. It is never a refutation of $\eta$, and the
temptation to read it as one is the reason the rule is written down.
\end{grindrule}

\begin{grindrule}[Push]
\[
  \frac{\Xi\vdash\eta:\mathsf{brk}(P)\ \text{with stated limitation }\ell}
       {\partial_\ell\eta\in\mathcal H}\;(\textsc{Push})
\]
Every limitation re-enters the hypothesis space as a target.
\deleted{The autonomous panel need not wait for final promotion before pushing a
reproducible candidate:}\changegap
\added{An admitted candidate may be refined before audit:}
\added{
\[
  \frac{\Xi\vdash\eta:\mathsf{ev}\ \text{with stated limitation }\ell}
       {\partial_\ell\eta\in\mathcal H}\;(\textsc{AutoPush}).
\]}
\deleted{This rule changes the next research question, not the evidentiary type
of $\eta$.}\changegap
\added{Thus autonomous search can attack a candidate's stated limitations
without assigning it a final type.}
\end{grindrule}

\begin{lemma}[Two Exits]
\label{l:twoexits}
A \emph{completed} push has two productive exits. Either $\partial_\ell\eta$ is
admitted, and the result strengthens; or $\partial_\ell\eta$ is refuted by a
proof that $\ell$ cannot be removed, and $\ell$ upgrades from an apology to a
maximality theorem. A push need not complete: by Rule~\ref{r:unr} it may return
$\mathord{?}$ and stand unresolved, with neither a stronger attack nor an
impossibility proof in hand. That is the ordinary outcome and is recorded as
such.
\end{lemma}

Lemma~\ref{l:twoexits} is why the push is worth running even when it is
expected to fail. The Type-I block maximality statement in
Section~\ref{sec:group} is
precisely a second exit: an attempt to widen a quotient attack that instead
proved the quotient could not be widened.

\subsection{The grind operator}

\begin{definition}[Grind]
Let $\Sigma_n\subseteq\mathcal H$ be the survivors after round $n$ and
$\Xi_n$ the evidence store. The grind operator is the composite
\[
  \Gamma \;=\; \mathfrak S\circ\mathfrak P\circ
               (\mathfrak X\sqcup\mathfrak E)\circ\mathfrak D,
  \qquad
  (\Sigma_{n+1},\Xi_{n+1})=\Gamma(\Sigma_n,\Xi_n).
\]
A run reaches \emph{grind normal form} at the least $n$ with
$\Sigma_{n+1}=\Sigma_n$: every surviving limitation has been attacked and has
moved, been proved immovable, or been left unresolved with the attempt
recorded. Normal form is a stopping condition, not a completeness claim.
\deleted{The operator $\Gamma$ describes the autonomous search only. It may
produce and recursively push candidate insights, but promotion under the
classification rules remains external to $\Gamma$ because it requires
$\mathrm{Aud}$.}\changegap
\added{Thus $\Gamma$ governs autonomous search; $\mathrm{Aud}$ and
Rule~\ref{r:brk} act on its output.}
\end{definition}

What survived this run is countable: eleven results, three coverage gaps, and
one rejected strengthening retained for its discriminating evidence. What was
proposed and failed was not counted, because the generator was never
instrumented to record it.

That asymmetry is the economics of the method rather than a gap in its
reporting. The survival rate could be made arbitrarily small by proposing more,
and the method would be no worse for it, because the cost of a failed
hypothesis is bounded by the cost of one exact check while the value of a
surviving one is not bounded at all. What matters is the size of
$\Sigma_\infty$, not the fraction of proposals reaching it.

\subsection{\texorpdfstring{\deleted{What this does not formalise}\changegap
\added{Scope of the calculus}}{Scope of the calculus}}

The calculus fixes two layers: where hypotheses come from, through the joins
and the two search directions, and \deleted{when one of them becomes evidence
and then a break}\deleted{when one becomes reproducible candidate evidence; the
human-audit predicate separately controls promotion}\changegap
\added{when a hypothesis is admitted as computational evidence and how an
audited hypothesis is classified}. It leaves a third layer
alone. Which joins in a given construction repay
attention, and in what order, is not derivable from anything above, and a run
that picks badly returns nothing while satisfying every rule here. $\mathfrak D$
supplies volume and $\mathsf A$ and $\mathsf B$ supply aim, and neither
supplies judgement about where to point them first.

\subsection{\texorpdfstring{\deleted{The same two diagrams, in the same overly
formal spirit}\changegap\added{Formal workflow and escalation ladder}}{Formal
workflow and escalation ladder}}

The two diagrams of Section~\ref{sec:grind} are repeated here with every arrow
named. Figure~\ref{fig:formalloop} is the loop of Figure~\ref{fig:loop} with
the agents, the admission premises, and the three verdicts labelled, so that
the door into evidence and the door into a classification appear as the
separate steps they are. Figure~\ref{fig:formalladder} is the ladder of
Figure~\ref{fig:ladder} as a chain in the escalation poset, carrying the
fibres that the earlier figure deliberately leaves out: there the rungs matter
and here their contents do.

\begin{figure}[!ht]
\centering
\resizebox{\linewidth}{!}{%
\begin{tikzpicture}[
  font=\small,
  st/.style={draw, rounded corners=2pt, inner sep=4.5pt, minimum height=7.5mm,
             minimum width=13mm},
  sink/.style={st, densely dotted},
  lbl/.style={font=\scriptsize, inner sep=2pt},
  >={Stealth[round,length=5pt]},
  every path/.style={thick}
]
\node[st] (J) at (0,0) {$\mathcal J(\pi)$};
\node[st] (H) at (3.1,0) {$\mathcal H$};
\node[st] (X) at (8.5,0) {$\mathfrak X$};
\node[st] (S) at (13.2,2.15) {$\Sigma$};
\node[st] (B) at (13.2,0) {$\bot$};
\node[sink] (U) at (8.5,-1.9) {unresolved};
\node[st] (E) at (16.8,2.15) {$\Xi$};

\draw[->] (J) -- node[lbl, above] {$\mathfrak D$} (H);
\draw[->] (H) -- node[lbl, above]
                 {\deleted{\textsc{Adm}}\changegap\added{$\mathfrak X$}}
                 node[lbl, below]
                 {\deleted{$\mathrm{Ex}\wedge\mathrm{Wit}\wedge\mathrm{Ctl}
                   \wedge\mathrm{Par}$}} (X);
\draw[->] (X) -- node[lbl, sloped, above, align=center]
                 {$\top$\added{$\,;\,\textsc{Adm}$}\\[1pt]
                  \added{\footnotesize$\mathrm{Ex}\wedge\mathrm{Wit}
                  \wedge\mathrm{Ctl}\wedge\mathrm{Par}$}} (S);
\draw[->] (X) -- node[lbl, sloped, above] {$\bot$}
                 node[lbl, sloped, below] {$\mathfrak E$} (B);
\draw[->, densely dotted] (X) -- node[lbl, right] {$\mathord{?}$} (U);
\draw[->] (S) -- node[lbl, above]
  {\added{$\mathrm{Aud}\,;$ }$\mathfrak S$} (E);

\draw[->] (S) to[out=125, in=55, looseness=0.9]
  node[lbl, pos=0.5, above=2.6mm] {$\mathfrak P:\eta\mapsto\partial_\ell\eta$} (H);
\draw[->, densely dashed] (B.south) -- ++(0,-2.75)
  -- node[lbl, above] {retained, not deleted} ($(H.south)+(0,-2.95)$)
  -- (H.south);
\end{tikzpicture}}
\caption{Figure~\ref{fig:loop} with the arrows named. $\mathfrak D$ proposes at
the joins, \deleted{\textsc{Adm} is the only door into evidence}\changegap
\deleted{\textsc{Adm} is the door into computational evidence and
$\mathrm{Aud}$ is the door from a candidate insight to a cryptanalytic claim}
\changegap\added{\textsc{Adm} records a reproducible candidate and
$\mathrm{Aud}$ precedes cryptanalytic classification},
and $\mathfrak P$
returns every stated limitation to $\mathcal H$. The $\bot$ branch is archived
and not discarded, so a rejected hypothesis keeps its discriminating
evidence and can be reopened. The third verdict is a sink: an unresolved path
is neither evidence nor a refutation, and nothing may be concluded from
reaching it.}
\label{fig:formalloop}
\end{figure}
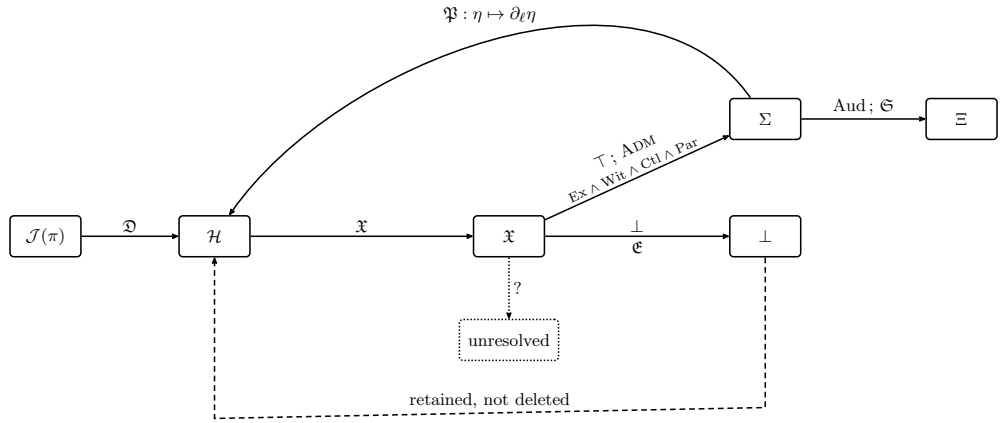

\begin{figure}[!ht]
\centering
\resizebox{\linewidth}{!}{%
\begin{tikzpicture}[
  font=\footnotesize,
  lv/.style={draw, rounded corners=2pt, align=center, inner sep=4pt,
             minimum height=8mm, text width=1.75cm},
  every node/.append style={inner ysep=3pt},
  top/.style={lv, line width=1.1pt},
  tag/.style={font=\scriptsize, inner sep=1.5pt},
  fib/.style={align=center, font=\scriptsize},
  >={Stealth[round,length=4pt]},
  every path/.style={thick}
]
\node[lv] (l1) {proof or spec flaw};
\node[lv, right=6mm of l1] (l2) {exact counter-example};
\node[lv, right=6mm of l2] (l3) {assumption break};
\node[lv, right=6mm of l3] (l4) {distinguisher};
\node[lv, right=6mm of l4] (l5) {key recovery};
\node[lv, right=6mm of l5] (l6) {plaintext recovery};
\node[top, right=6mm of l6] (l7) {security game broken};
\foreach \x/\y in {l1/l2,l2/l3,l3/l4,l4/l5,l5/l6,l6/l7}
  {\draw[->] (\x) -- node[tag, above=0.6mm] {$\prec$} (\y);}

\node[tag, above=2.2mm of l1] {$\lambda_1$};
\node[tag, above=2.2mm of l2] {$\lambda_2$};
\node[tag, above=2.2mm of l3] {$\lambda_3$};
\node[tag, above=2.2mm of l4] {$\lambda_4$};
\node[tag, above=2.2mm of l5] {$\lambda_5$};
\node[tag, above=2.2mm of l6] {$\lambda_6$};
\node[tag, above=2.2mm of l7] {$\lambda_7$};

\node[fib, below=4.5mm of l1, text width=2.0cm]
  {$\mu^{-1}(\lambda_1)$\\Hint-MLWE\\Hollow params\\MinRank};
\node[fib, below=4.5mm of l3, text width=2.0cm]
  {$\mu^{-1}(\lambda_3)$\\group-ring\\quotients};
\node[fib, below=4.5mm of l5, text width=2.0cm]
  {$\mu^{-1}(\lambda_5)$\\function-field\\basis};
\node[fib, below=4.5mm of l7, text width=2.2cm]
  {$\mu^{-1}(\lambda_7)$\\Ring-LWR\\middle-product\\e-voting\\Hollow-LWE\\\added{Spinel}\\CRT-RLWE};
\end{tikzpicture}}
\caption{Figure~\ref{fig:ladder} as a chain
$\lambda_1\prec\cdots\prec\lambda_7$ in the escalation poset $\Lambda$, with
the attainment map
\deleted{$\mu:\Sigma_\infty\to\Lambda$}\changegap
\added{$\mu:\Sigma_\infty^{\mathrm{Aud}}\to\Lambda$, where
$\Sigma_\infty^{\mathrm{Aud}}=\{\eta\in\Sigma_\infty:\mathrm{Aud}(\eta)\}$,}
sending each audited survivor to the
highest level its evidence reaches. The fibres $\mu^{-1}(\lambda_i)$ are the
contents of each level. A finding halts at $\mu(\eta)=\lambda_i$ exactly when
$\neg\,\mathrm{Brg}$ holds for every $P$ at $\lambda_{i+1}$, or when
Lemma~\ref{l:twoexits} has taken its second exit and the halt is itself a
theorem. \deleted{The domain of $\mu$ contains only findings that have passed
the human audit predicate.}}
\label{fig:formalladder}
\end{figure}
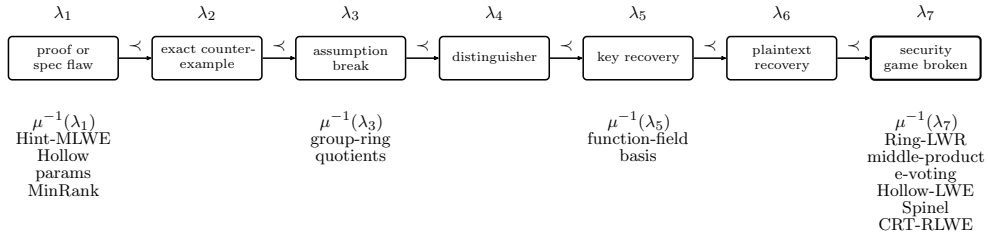

\end{document}